\documentclass{article}
\usepackage{amsmath, amsfonts, amssymb}
\usepackage{amsthm,comment,color,latexsym,graphicx,url,enumitem}
\usepackage[margin=1in]{geometry}
\usepackage[font=small,labelfont=bf]{caption}
\usepackage[labelformat=simple]{subcaption}

\usepackage{libertine}\usepackage[libertine]{newtxmath}
\usepackage[scaled=0.96]{zi4}

\usepackage[utf8]{inputenc}

\newif\ifabstract
\abstracttrue
\newif\iffull
\ifabstract \fullfalse \else \fulltrue \fi

\usepackage{hyperref}
\hypersetup{breaklinks,bookmarks,bookmarksnumbered,bookmarksopen,bookmarksopenlevel=2}
{\makeatletter \hypersetup{pdftitle={\@title}}}

{\makeatletter
 \gdef\xxxmark{%
   \expandafter\ifx\csname @mpargs\endcsname\relax 
     \expandafter\ifx\csname @captype\endcsname\relax 
       \marginpar{xxx}
     \else
       xxx 
     \fi
   \else
     xxx 
   \fi}
 \gdef\xxx{\@ifnextchar[\xxx@lab\xxx@nolab}
 \long\gdef\xxx@lab[#1]#2{\textbf{[\xxxmark #2 ---{\sc #1}]}}
 \long\gdef\xxx@nolab#1{\textbf{[\xxxmark #1]}}
}

{\makeatletter \gdef\fps@figure{!htbp}}

\let\realbfseries=\bfseries
\def\bfseries{\realbfseries\boldmath}

\newtheorem{theorem}{Theorem}[section]
\newtheorem{lemma}[theorem]{Lemma}
\newtheorem{proposition}[theorem]{Proposition}

\newtheorem{claim}[theorem]{Claim}

\theoremstyle{definition}
\newtheorem{definition}[theorem]{Definition}
\newtheorem{openprob}[theorem]{Open Problem}

\theoremstyle{remark}
\newtheorem{remark}[theorem]{Remark}
\let\epsilon=\varepsilon
\def\defn#1{\textbf{\textit{\boldmath #1}}}

\usepackage{asymptote}

\newcommand{\vocab}[1]{\textbf{\textit{#1}}}

\usepackage{graphicx}
\usepackage{caption}
\usepackage{subcaption}
\usepackage{hyperref}
\title{Complexity of 2D Snake Cube Puzzles}

\author{%
  MIT Hardness Group%
    \thanks{Artificial first author to highlight that the other authors (in
      alphabetical order) worked as an equal group. Please include all
      authors (including this one) in your bibliography, and refer to the
      authors as “MIT Hardness Group” (without “et al.”).}
\and
  Nithid Anchaleenukoon%
    \thanks{MIT Computer Science and Artificial Intelligence Laboratory,
      32 Vassar St., Cambridge, MA 02139, USA, \protect\url{{nithidan,alexdang,edemaine,kayleeji,psaeng}@mit.edu}}
\and
  Alex Dang\footnotemark[2]
\and
  Erik D. Demaine\footnotemark[2]
\and
  Kaylee Ji\footnotemark[2]
\and
  Pitchayut Saengrungkongka\footnotemark[2]
}

\date{}

\begin{document}
\begin{asydef}
	settings.outformat = "pdf";
	texpreamble("\usepackage{libertine} \usepackage[libertine]{newtxmath} \usepackage[scaled=0.96]{zi4} \usepackage[utf8]{inputenc}");
\end{asydef}
\maketitle

\begin{abstract}
  Given a chain of $HW$ cubes where each cube is marked ``turn $90^\circ$'' or ``go straight'', when can it fold into a $1 \times H \times W$ rectangular box?
  We prove several variants of this (still) open problem NP-hard:
  (1)~allowing some cubes to be wildcard (can turn or go straight);
  (2)~allowing a larger box with empty spaces
  (simplifying a proof from CCCG 2022);
  (3)~growing the box (and the number of cubes) to $2 \times H \times W$
  (improving a prior 3D result from height $8$ to~$2$);
  (4)~with hexagonal prisms rather than cubes, each specified as going straight,
  turning $60^\circ$, or turning $120^\circ$; and
  (5)~allowing the cubes to be encoded implicitly to compress exponentially large repetitions.
\end{abstract}

\section{Introduction}

\vocab{Snake Cube} \cite{abel2013finding} is a physical puzzle consisting of wooden unit cubes joined in a chain by an elastic string running through the interior of each cube. For every cube other than the first and last, the string constrains the two neighboring cubes to be at opposite or adjacent faces of this cube, in other words, whether the chain must continue straight or turn at a $90^\circ$ angle. In the various manufactured puzzles, the objective is to rearrange a chain of $27$ cubes into a $3 \times 3 \times 3$ box.

To generalize this puzzle, we ask: given a chain of $DHW$ cubes, where $D, H, W$ are positive integers, is it possible to rearrange the cubes to form a $D \times H \times W$ rectangular box? We call this problem $D \times H \times W$ \textsc{Snake Cube}.
Previous results on its complexity include:
\begin{itemize}[itemsep=0pt]
    \item Abel et al.~\cite{abel2013finding} proved $8 \times H \times W$ \textsc{Snake Cube} is NP-complete by reduction from \textsc{$3$-Partition}.
    \item Demaine et al.~\cite{original2Dsnake} proved \textsc{2D Snake Cube Packing}---deciding whether a chain of cubes can \vocab{pack} (but not necessarily fill) a $1 \times H \times W$ rectangular box where all cubes are constrained to align with the box---is NP-complete by reduction from \textsc{Linked Planar 3SAT}. This result also holds for a closed chain \cite{original2Dsnake}.

    This result improves a previous result by Demaine and Eisenstat \cite{fixedangle_chain}, which proves NP-hardness of the problem that given a sequence of angles in interval $[16.26^\circ, 180^\circ]$, is there a non-crossing polyline of side-length 1 and the angle sequence as given. \cite{original2Dsnake} does this with angles only in $\{90^\circ, 180^\circ\}$.
\end{itemize}
Both \cite{abel2013finding} and \cite{original2Dsnake} pose the (still) open problem of determining the complexity of $1 \times H \times W$ \textsc{Snake Cube},
illustrated in Figure~\ref{fig:example_snake}:
\begin{openprob}[2D Snake Cube]
\label{prob:2d_snake}
    Is $1\times H \times W$ \textsc{Snake Cube} NP-hard?
\end{openprob}

\subsection{Our Results}

In this paper, we prove NP-hardness of several variations of Open Problem \ref{prob:2d_snake}:
\begin{itemize}
\item In Section \ref{sec:2D_wildcard}, we prove NP-completeness of \textsc{2D Snake Cube with Wildcards}, where some instructions can be the wildcard \texttt{*}, allowing a free choice between straight or turn.

We also give an alternative proof that \textsc{2D Snake Cube Packing} is NP-complete, simplifying \cite{original2Dsnake}.
\item In Section \ref{sec:2xmxn_fill}, we prove that $2 \times H \times W$ \textsc{Snake Cube} is NP-complete. This improves the result of Abel et al.~\cite{abel2013finding} from $D=8$ to $D=2$.
\item In Section \ref{sec:triangular}, we prove NP-completeness of \textsc{Hexagonal 2D Snake Cube Packing}: deciding whether a chain of hexagonal prisms each specified as going straight, turning $60^\circ$, or turning $120^\circ$ can be packed into a $60^\circ$, $H \times W$ parallelogram. Similar to \cite{original2Dsnake}, we extend this result to closed chains. One can view this as an improvement to \cite{fixedangle_chain} in that angles can be restricted to be in $\{60^\circ, 120^\circ\}$.
\item In Section \ref{sec:weak2DFill}, we prove \emph{weak} NP-hardness of \textsc{2D Snake Cube}, when the input instructions can be encoded to efficiently represent repeated sequences.
\end{itemize}

The first three results are all based on a common framework for reduction from \textsc{Numerical 3D Matching}, detailed in Sections~\ref{sec:overview_3DM} and~\ref{sec:wirepacking}, while the last result is a reduction from \textsc{2-Partition}. 
We introduce the two base problems we reduce from, and formally define the problems we prove hard, in Section~\ref{sec:preliminaries}.

\section{Definitions}
\label{sec:preliminaries}

\subsection{Problem Statement}
First, we define the problems we analyze in more precise terms. 

The input to all snake-cube problems is a ``box'' and a ``program'';
refer to Figure~\ref{fig:example_snake}.
A \vocab{box} is a $D \times H \times W$ rectangular cuboid that the cubes of the snake-cube puzzle must fit into. This box can be visualized as a cubic grid where each cube occupies one space of the grid.
A \vocab{program} $\mathcal{P} = (p_1,p_2, \ldots, p_k)$ is a length-$k$ string of \vocab{instructions}, where each instruction $p_i$ is either the character $\texttt S$ (``straight'') or $\texttt T$ (``turn'').
The corresponding \vocab{chain} is the sequence $c_1, c_2, \dots, c_k$ of pairwise distinct and consecutively adjacent cubes that must \defn{follow the program} in the sense that each instruction $p_i$ (where $i \in \{2, \ldots, k-1\}$) constrains the angle between the three cubes $c_{i-1}, c_i, c_{i+1}$ to be $180^{\circ}$ for $p_i = \texttt S$ (i.e., going straight), and $90^{\circ}$ for $p_i = \texttt T$ (i.e., turning with \emph{absolute} turn angle $90^\circ$).
(Note that the first and last instructions have no effect; we include them so that cubes and instructions have a bijection.)
A length-$k$ \vocab{segment} refers to a subchain of $k$ consecutive cubes that are constrained to form a straight line (e.g., the subchain following the instructions $\texttt T \texttt S \texttt S \texttt S \texttt T$ is a length-$5$ segment).

\begin{figure}[t]
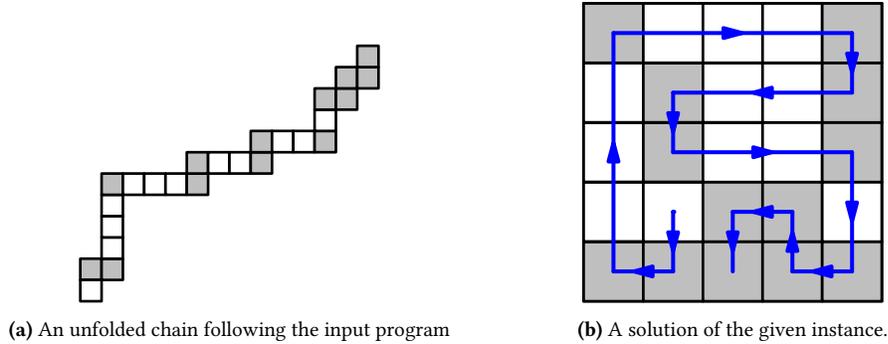

    \centering
    \begin{subfigure}{0.5\textwidth}
        \centering
        \asyinclude{asy/example_unfolded.asy}
        \caption{An unfolded chain following the input program}
        \label{fig:unfolded_example}
    \end{subfigure}
    \begin{subfigure}{0.3\textwidth}
        \centering
        \asyinclude{asy/example_folded.asy}
        \caption{A solution of the given instance.}
        \label{fig:folded_example}
    \end{subfigure}
    
    \caption{An example instance of \textsc{$D\times H\times W$ Snake Cube} with input $D = 1$, $H = 5$, $W = 5$, and the program $\texttt{STTSSSTSSSTTSSTTSSTSTTTTT}$. The gray and white cells represent cubes following instructions $\texttt T$ and $\texttt S$, respectively.}
    \label{fig:example_snake}
\end{figure}
\begin{definition} \textsc{$D\times H\times W$ Snake Cube} is the problem of deciding, given positive integers $D$, $H$, $W$ (defining the box) and program $\mathcal P$ of length $DHW$, whether there is a chain of cubes following the program $\mathcal P$. \textsc{2D Snake Cube} refers to the special case of $D=1$.
\end{definition}
\begin{definition} \textsc{2D Snake Cube Packing} is a variant of \textsc{2D Snake Cube} where $\mathcal P$ may have length less than $DHW$. Thus, the chain does not necessarily include all cubes in the box, but must stay within the box.
\end{definition}
\begin{definition} \textsc{2D Snake Cube with Wildcards} is a variant of \textsc{2D Snake Cube} in which $\mathcal P$ may contain a third character \texttt{*}, where $p_i = \texttt{*}$ constrains the angle between the three cubes $c_{i-1}, c_{i}, c_{i+1}$ to be either $90^\circ$ or $180^\circ$.
\end{definition}
We also consider \textsc{Hexagonal 2D Snake Cube Packing}, a version of 2D Snake Cube on hexagonal prisms with three different characters: one constraining the chain to go straight, one constraining the chain to turn by $\pm 60^{\circ}$, and one constraining the chain to turn by $\pm 120^{\circ}$. We more formally define this in Section \ref{sec:triangular}.

We use the notation $(s)^k$ to denote the instruction sequence $s$ repeated $k$ times. For example, $\texttt{T}(\texttt{SST})^2$ represents the sequence $\texttt{TSSTSST}$. For almost all our proofs, this notation is just convenient shorthand. The exception is our proof that \textsc{2D Snake Cube} is weakly NP-hard, which requires the input instruction to be compressed using this encoding, i.e., encode $(s)^k$ using $|s| + O(\log k)$ bits.

All of these problems, with the exception of \textsc{2D Snake Cube} under the special encoding, are in NP: to verify a solution, simply check all angle constraints (of which there are linearly many) and that the chain remains within the box. This algorithm is linear in both time and space. Thus, showing any of these problems is NP-hard suffices to establish NP-completeness.

\subsection{2-Partition}

Our weak NP-hardness of \textsc{2D Snake Cube} (in Section \ref{sec:weak2DFill}) will reduce from \textsc{2-Partition}, defined as follows:
\begin{definition}[2-Partition]
    Given a multiset $A = \{a_1, a_2, \ldots, a_n\}$ of positive integers, \textsc{2-Partition} is the problem of deciding whether there exists a partition of sets $A$ into disjoint union $A_1\sqcup A_2$ such that the sums of elements in $A_1$ and in $A_2$ are equal.
\end{definition}
We will need the following known result:
\begin{theorem}
    \textsc{2-Partition} is \defn{weakly NP-hard}, i.e., NP-hard when the numbers $a_i$ are encoded in binary (and thus could be exponential in value).
\end{theorem}
\begin{proof}
    This is listed as Problem SP12 in Garey and Johnson's compilation \cite[Section A3.2]{garey_johnson}.
\end{proof}

\subsection{Numerical 3D Matching}
Our other (strong) NP-hardness proofs are reductions from \textsc{Numerical 3D Matching}, defined as follows:
\begin{definition}
For any given target sum $t$ and sequences $(a_i)_{i=1}^n$, $(b_i)_{i=1}^n$, and $(c_i)_{i=1}^n$, each consisting of $n$ positive integers, \vocab{Numerical 3D Matching (Numerical 3DM)} is the problem of deciding
whether there exist permutations $\sigma$ and $\pi$ of the set $\{1,2,\dots,n\}$ that satisfies $a_i + b_{\sigma(i)} + c_{\pi(i)} = t$ for all $i$. We refer to such a solution $(\sigma,\pi)$ as a \vocab{matching}.
\end{definition}
\begin{theorem}
   \textsc{Numerical 3DM} is \defn{strongly NP-hard}, i.e., NP-hard when the numbers $a_i,b_i,c_i,t$ are encoded in unary (and thus are polynomial in value).
\end{theorem}
\begin{proof}
     This is listed as Problem SP16 in Garey and Johnson's compilation \cite[Section A3.2]{garey_johnson}.
\end{proof}

In fact, we will consider the following slightly modified version of 3D Matching.
\begin{proposition}
    \label{prop:3DM_tweak}
    \textsc{Numerical 3DM} is strongly NP-hard even when we assume that
    $a_i\in (0.5t, 0.6t)$, $b_i\in (0.25t, 0.3t)$, and
    $c_i\in (0.125t, 0.15t)$ for all $1\leq i\leq n$.
\end{proposition}
\begin{proof}
    Suppose that we have an instance of Numerical 3DM $a_1,\dots,a_n$,
    $b_1,\dots,b_n$, $c_1,\dots,c_n$ with target sum $t$.
    We can transform the instance by defining a variable $X$ and setting
    $$a_i' = a_i + 4X,\qquad b_i' = b_i + 2X,\qquad c_i' = c_i+X,
    \qquad t' = t + 7X.$$
    
    Intuitively, when $X$ is sufficiently large,
    we have $t'\approx 7X$, so $a_i'\approx \frac 47 t'$,
    $b_i'\approx\frac 27 t'$, and $c_i'\approx\frac 17t'$,
    which makes all inequalities hold.
    To be more precise, we claim that, for any $X$ larger than a bound
    depending polynomially on $a_i$, $b_i$, $c_i$, and $n$,
    we have
    $$a_i' \in (0.5t', 0.6t'),\qquad
    b_i' \in (0.25t', 0.3t'),\quad\text{ and }\quad
    c_i' \in (0.125t', 0.15t').$$
    To see why, note that in order for the first inequality
    to hold, we must have
    $$a_i + 4X > 0.5(t+7X) \iff X > t-2a_i
    \quad\text{ and } \quad
    a_i + 4X < 0.6(t+7X) \iff X > 5a_i - 3t.$$
    Solving the other two inequalities, we find that we need
    $$X > \max\{t-4b_i, 10b_i-3t, t-8c_i, 20c_i-3t\}.$$
    Thus, we only need $X$ to be greater than the maximum of six linear expressions in $a_i$, $b_i$, $c_i$, and $t$.
    Taking the maximum across all $i$, we get that the working $X$ has polynomial size as claimed.
\end{proof}

\section{Overview of Reductions from Numerical 3DM}
\label{sec:overview_3DM}
The reductions in Sections \ref{sec:2D_wildcard}, \ref{sec:2xmxn_fill}, and \ref{sec:triangular}
all share a similar infrastructure, which we informally outline here. In this overview, we assume $D = 1$. We explain how to adapt this framework to $D = 2$ in Section \ref{sec:2xmxn_fill}.

We reduce from \textsc{Numerical 3DM}. Let $(a_i)_{i=1}^n$, $(b_i)_{i=1}^n$, and $(c_i)_{i=1}^n$ be the input numbers to match with target sum $t$. By Proposition \ref{prop:3DM_tweak}, we can assume that $a_i\in (0.5t, 0.6t)$, $b_i\in (0.25t, 0.3t)$, and $c_i\in (0.125t, 0.15t)$.
We also pick the following parameters:
\begin{align*}
    g &= \Theta(n) && \text{(the gap size)}, \\
    h &= \Theta(n^2) && \text{(the height of each block)}, \\
    m &= \Theta(n^3) && \text{(the unit width of each block)}.
\end{align*}

The structure of the reduction is as follows. The dimensions of the box are $D\times H\times W = 1\times (nh+(n+1)g)\times (mt+4g)$. The numbers $(a_i)_{i=1}^n$, $(b_i)_{i=1}^n$, and $(c_i)_{i=1}^n$ are represented by \vocab{block gadgets} $(\langle A_i \rangle)_{i=1}^n$, $(\langle B_i\rangle)_{i=1}^n$, and $(\langle C_i \rangle)_{i=1}^n$, instructions corresponding to the subchains $(A_i)_{i=1}^n$, $(B_i)_{i=1}^n$, and $(C_i)_{i=1}^n$ that can fold into cuboids of dimensions $1 \times h \times ma_i$, $1 \times h \times mb_i$, and $1 \times h \times mc_i$, respectively. We call these subchains \vocab{blocks}. Block gadgets typically consist of $h$ long segments (of common length $m x$) interspersed by $h-1$ length-$2$ segments, as shown in Figure \ref{subfig:block}, but details vary in different variants. Between every two consecutive block gadgets, we add a \vocab{wiring gadget}, a sequence of instructions that allows connecting the two blocks no matter where they are in the box. 

\begin{figure}
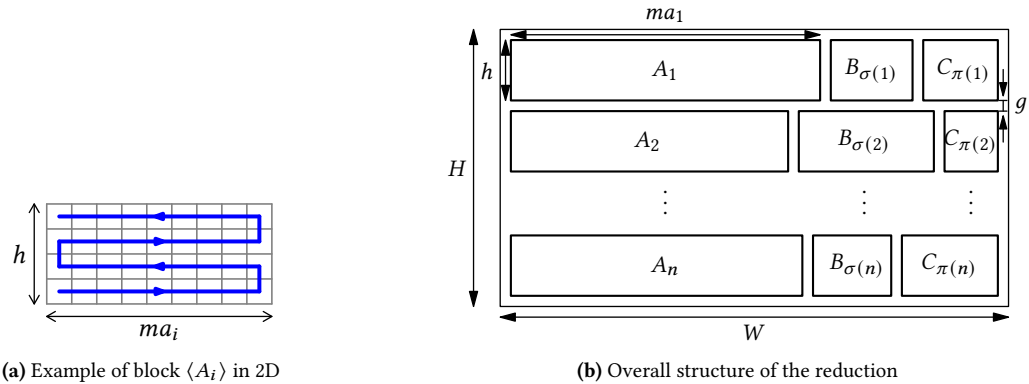

    \centering
    \begin{subfigure}{0.25\textwidth}
        \centering
        \asyinclude{asy/single_block.asy}
        \caption{Example of block $\langle A_i\rangle$ in 2D}
        \label{subfig:block}
    \end{subfigure}
    \begin{subfigure}{0.7\textwidth}
        \centering
        \asyinclude{asy/reduction_outline.asy}
        \caption{Overall structure of the reduction}
        \label{subfig:reduction_outline}
    \end{subfigure}
        \caption{The reduction}
\end{figure}

If a matching exists (i.e., there exist two permutations $\sigma$ and $\pi$ of $\{1, 2, \ldots, n\}$ such that $a_i + b_{\sigma(i)} + c_{\pi(i)}=t$ for all~$i$), then (ignoring the wiring gadgets) the blocks can be arranged into a perfect $1 \times nh\times mt$ rectangle, by aligning each triple of blocks $A_i$, $B_{\sigma(i)}$, and $C_{\pi(i)}$ together in the same row. Refer to Figure~\ref{subfig:reduction_outline}. Because our box is slightly larger than $1 \times nh\times mt$, we can place the blocks with a gap of $g$ between neighboring blocks and between each block and the boundary of the rectangular box. The gap $g$ is chosen so there is sufficient space for a subchain following the wiring gadget to connect all the blocks. Wires detour around blocks and do not cross; the explicit algorithm for folding the subchain will be given in Lemma~\ref{lem:wire}. Finally, depending on the variant, there may be additional instructions at the end of the program corresponding to the cubes needed to fill the remaining space in the box.

In the other direction, we also need to show that the existence of a chain satisfying the instructions forces the existence of a corresponding matching. One might expect $\langle A_i\rangle$, $\langle B_i\rangle$, and $\langle C_i\rangle$ to fold into blocks and pack into the grid equally-spaced apart, which will immediately force the existence of a matching. However, there are two major difficulties.
\begin{itemize}
    \item First, part of the chain corresponding to $\langle B_i\rangle$ and $\langle C_i\rangle$ can slide between the blocks $A_i$s, causing various unexpected configurations.
    \item Even if the blocks do fold into perfect rectangles, one must rule out the case when it is packed with unequal spacing between blocks.
\end{itemize} 
To get around these issues, we will prove Lemma \ref{lem:segment_packing} that will force the existence of matching even if blocks do not fold ideally.

\section{Common Results}
\label{sec:wirepacking}
In this section, we will collect arguments that are shared across different reductions from \textsc{Numerical 3DM}.
\subsection{Segment Packing Lemma}
This subsection concerns the part ``Chain $\Rightarrow$ Matching'' (i.e., proving the existence of Numerical 3DM matching).
As said earlier in Section \ref{sec:overview_3DM}, we need to force the existence of a matching even when blocks do not fold into ideal rectangles or are unevenly spaced. To do so, we will prove the ``Segment Packing Lemma'' (Lemma \ref{lem:segment_packing}), which will be used in Sections \ref{sec:2D_wildcard} and \ref{sec:triangular}. 

We first explain the motivation of this lemma. Recall a block gadget typically consists of $h$ long segments interspersed with $h-1$ length-2 segments. Therefore, for instance, any block $A_i$ must have $h$ segment of length $ma_i$. The idea is to view a block as a collection $h$ horizontal segments. 
Therefore, in any chain satisfying the program in the outline, we must have $3nh$ segments, $h$ of each length $ma_1,\dots,ma_n$, $mb_1,\dots,mb_n$, $mc_1,\dots,mc_n$ that were packed in a grid. This motivates the ``Segment Packing Lemma''.

The setup of this lemma is as follows.
Let $a_1,\dots,a_n$, $b_1,\dots,b_n$, $c_1,\dots,c_n$ be positive integers summing to $nt$. Assume the condition in Proposition \ref{prop:3DM_tweak}: for all $i$, we have
$$a_i\in (0.5t, 0.6t),\qquad b_i\in (0.25t, 0.3t),\quad\text{ and } \quad c_i\in (0.125t, 0.15t).$$ Let $m$ and $h$ be positive integers such that $m > H$, and consider a $H\times W$ rectangle.
Inside this rectangle, there are $3nh$ segments, each of which is horizontal. The segments are classified into $3n$ types:
\begin{itemize}
    \item $n$ types of \vocab{A-segments}: types $A_1$, $A_2$, \dots, $A_n$;
    \item $n$ types of \vocab{B-segments}: types $B_1$, $B_2$, \dots, $B_n$; and
    \item $n$ types of \vocab{C-segments}: types $C_1$, $C_2$, \dots, $C_n$.
\end{itemize}
\begin{lemma}[Segment Packing Lemma]
    \label{lem:segment_packing}
    Suppose that along with our setup, we also have that
    $$W < m(t+1)\quad\text{and}\quad H < nh + \frac{h}{40}$$
    and all of the following are true:
    \begin{itemize}
    \item There are exactly $h$ segments of each type.
    \item For all $1\leq i\leq n$, all segments of type $A_i$, $B_i$, and $C_i$ have lengths $ma_i$, $mb_i$, and $mc_i$, respectively.
    \item No two segments of the same type are more than $h$ rows vertically apart.
    \end{itemize}
    Then, there exists a matching of the instance $(a_i)_{i=1}^n$, $(b_i)_{i=1}^n$, and $(c_i)_{i=1}^n$ of Numerical 3DM.
\end{lemma}
The rest of the subsection is devoted to proving the lemma. Assume the lemma setup throughout.
First, we make the following observations.
\begin{proposition} 
\label{prop:row_properties}
We have the following.
\begin{enumerate}[label=(\roman*)]
\item Any row must have at most one A-segment. 
\item If a row contains one A-segment, it must contain at most one B-segment.
\item If a row contains one A-segment and one B-segment, then it must contain at most one C-segment. 
\end{enumerate}
\end{proposition}
\begin{proof}
Note that the length of an A-segment is $ma_i \geq 0.5m(t+1) > 0.5W$, the length of a B-segment is $mb_i \geq 0.25m(t+1) > 0.25W$, and the length of a C-segment is $mc_i \geq 0.125m(t+1) > 0.125W$. With these, we now prove each part as follows.
\begin{enumerate}[label=(\roman*)]
\item More than one A-segment will span a width exceeding $W$, a contradiction.
\item If a row contains one A-segment and more than one $B$-segment, the length will be at least $ma_i + 2mb_i > 0.5W + 2 \cdot 0.25W = W$, a contradiction.
\item If a row contains one A-segment, one $B$-segment, and more than one $C$-segment, the length will be at least $ma_i + mb_i +2mc_i > 0.5W + 0.25W + 2 \cdot 0.125W = W$, a contradiction.
\qedhere
\end{enumerate}
\end{proof}
From Proposition \ref{prop:row_properties}, each row in the grid must be of one of the following four categories.
\begin{itemize}
    \item A \emph{good row}, which contains exactly one A-segment, one B-segment, and one C-segment.
    \item An \emph{A-bad row}, which contains no A-segment.
    \item A \emph{B-bad row}, which contains one A-segment but contains no B-segment.
    \item A \emph{C-bad row}, which contains one A-segment, one B-segment, but no C-segment.
\end{itemize}
Let $n_\text{good}$, $n_A$, $n_B$, and $n_C$ denote the number of good rows, $A$-bad rows, $B$-bad rows, and $C$-bad rows, respectively.
The critical claim is the following.
\begin{claim}
    \label{claim:few_bad_rows}
    $n_A+n_B+n_C < h$.
\end{claim}
\begin{proof}
    First, we note that there are $H < nh + h/40$ rows. Of these, $nh$ must contain an $A$-path, so 
    $$n_A < \frac{h}{40}.$$

    Next, there are $nh$ B-segments. The good rows and C-bad rows contain at most one B-segment. Moreover, since the length of each B-segment is strictly greater than $0.25W$, each A-bad row contains at most $3$ B-segments. Thus,
    \begin{align*}nh \leq 3n_A + n_C + n_{\text{good}} \implies
     n_B &\leq 2n_A + (n_A+n_B+n_C+n_{\text{good}}) - nh \\
     &= 2n_A + \frac{h}{40}  < \frac{3h}{40}.
     \end{align*}

     Finally, there are $nh$ C-segments. The good rows contain at most one C-segment. Moreover, the length of each C-segment is strictly greater than $0.125W$, so each A-bad row contains at most $7$ C-segments and each B-bad row contains at most $3$ C-segments. Thus,
     \begin{align*}nh \leq 7n_A + 3n_B + n_{\text{good}} \implies
     n_C &\leq 6n_A + 2n_B + (n_A+n_B+n_C+n_{\text{good}}) - nh \\
     &= 6n_A + 2n_B + \frac{h}{40}  < \frac{13h}{40}.
     \end{align*}
     
     Summing these three bounds gives the claim.
\end{proof}
Finally, we are ready to prove that the solution to an instance of Numerical 3DM exists. 
\begin{proof}[Proof of Lemma \ref{lem:segment_packing}]
Color each row by its residue modulo $h$. Thus, each color has either $n$ or $n+1$ rows. By Claim~\ref{claim:few_bad_rows}, there are less than $h$ bad rows, so there is a color $c$ such that all rows of that color are good. Moreover, each block gadget spans the height at most $h$, so each block covers at most one row of color $c$. Thus, there are exactly $n$ rows of color $c$. For each row, if that row contains a path from blocks $\langle A_i\rangle$, $\langle B_j\rangle$, $\langle C_k\rangle$, then by counting the lengths of the path, we have
$$ma_i + mb_j + mc_k \leq W < m(t+1) \implies a_i+b_j+c_k\leq t.$$
Summing this inequality for each row gives $nt \leq nt$, so all inequalities must be equalities, giving that $a_i+b_j+c_k=t$, forming the desired matching.
\end{proof}
\subsection{Connecting Wires}
This subsection concerns the wiring part and is the most technical part of this paper. It guarantees that, if the gap is large enough, there exists a way to wire blocks $A_1,\dots,A_n$, $B_1,\dots,B_n$ and $C_1,\dots,C_n$, no matter what positions they are placed in the packing. This lemma was adapted from \cite[Lemma 5]{foldingpolyhedra}. The existence will be proved by providing an algorithm to draw wires.

\begin{figure}[htp]
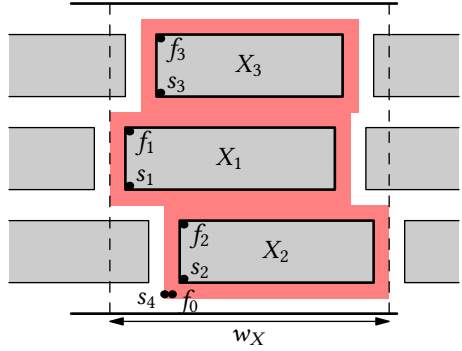

    \centering
    \asyinclude{asy/wiring_setup.asy}
    \caption{Example of the setup wire packing when $n=3$. The red area represents the available space.}
    \label{fig:wiresetup}
\end{figure}

The setup for this lemma is depicted in Figure \ref{fig:wiresetup} and goes as follows:
in a grid of \vocab{squares}, there are \vocab{rectangles} $X_1,X_2,\dots,X_n$, of height $h'$ squares
and widths $x_1,x_2,\dots,x_n$ squares packed in a rectangular $H'\times w_X$ grid (height $H'$). Note that while ``rectangles'' and ``blocks'' refers to the same object, ``squares'' and ``cubes'' are not the same; when applying this lemma to prove the construction of the wiring gadgets, each square is filled with $2\times 2$ cubes. This will be discussed further in Section~\ref{sec:2D_wildcard}. We use these terms here to avoid confusion when referring to the length. For example, if a rectangle has height of $h'$ squares in this section, when applying to the actual instance of \textsc{2D snake cube}, a corresponding ``block'' would have height $h = 2h'$ cubes.

A matching forces each row to contain at most one rectangle (i.e., rectangles are not overlapped and cannot be placed side by side), but they are in arbitrary order from top to bottom. 
Suppose the gap between two rectangles has width $g'$ squares. Define \vocab{available space} of a rectangle to be a set of squares outside it that forms an extension to a rectangle by $g'/2$ squares on each side. An \vocab{overall available space} is the union of all available spaces across all rectangles, which is shown as the red area in Figure~\ref{fig:wiresetup}.
For each $i$, let $s_i$ and $f_i$ be the squares at the bottom left and top left corners of $X_i$, respectively. We also let $f_0$ and $s_{n+1}$ be the beginning and the ending point in the available space lying below every rectangle at least $4n+4$ squares apart; they are both horizontally at least $10n$ tiles away from the left edges of all rectangles. 

For each $i\in\{0,1,\dots,n\}$, we want to connect from square $f_i$ to square $s_{i+1}$ with a wire $W_i$
of length $\ell_i$, defined as a path of $\ell_i$ squares such that any two consecutive squares share a common side.
This is guaranteed to be possible if the gap between rectangles is large enough, as we explain below.
\begin{lemma}[Wire Lemma]
\label{lem:wire}
Assume the setup in the above three paragraphs. Assume that $\min_i x_i > w_X/2$, $H < w_X$ and the gaps between rectangles have width at least $g'\geq 100n$ squares. For each $i=0,1,\dots,n$, let $\ell_i$ be an even integer in $[8nw_A, 12nw_A]$. Then, one can draw $n+1$ disjoint wires $W_0,\dots,W_n$ in the overall available space, where $W_i$ has length exactly $\ell_i$
and $W_i$ connects $f_i$ to $s_{i+1}$
for all $i\in\{0,1,\dots,n\}$.
Furthermore, no two cells in different wires $W_i$ and $W_j$ are adjacent.
\end{lemma}

\begin{remark}
In Sections \ref{sec:2D_wildcard} and \ref{sec:triangular}, we can set $\ell_i = 8nw_X$. 
For Section \ref{sec:2xmxn_fill}, we need to set $\ell_i = 12nw_X$ since wires have to detour around the ``shelf", an additional structure that will be defined there.
\end{remark}
\begin{proof}

We mark squares $m_0=f_0, m_1,\dots,m_n,m_{n+1}=s_{n+1}$, all in the same row, with distance at least $4$ apart from right to left. The square $m_i$ will be in the middle of the wire $W_i$. Thus, it suffices to draw $2n$ wires: $U_i$ from $m_{i-1}$ to $s_i$ and $V_i$ from $m_i$ to $f_i$ for each $i=1,2,\dots,n$. After that, we can set $W_i$ to be concatenation of $V_i$, a square $m_i$, and $U_{i+1}$ for all $i\in\{1,\dots, n-1\}$. Moreover, let $W_0 = U_1$ and $W_n = V_n$. Inductively running $i$ from $1,2,\dots,n$, we will draw $U_i, V_i$ after $U_1,V_1,\dots,U_{i-1},V_{i-1}$ has been drawn. The following diagram illustrates the notion described.
\[
\underbrace{f_0 = m_0 \xrightarrow{U_1}}_{W_0} \langle X_1\rangle \underbrace{\xrightarrow{V_1} m_1
 \xrightarrow{U_2} }_{W_1}\langle X_2\rangle \xrightarrow{V_2} m_2 \cdots 
  \xrightarrow{U_i} \langle X_i\rangle\underbrace{ \xrightarrow{V_i} m_i   \xrightarrow{U_{i+1}} }_{W_i} \langle X_{i+1}\rangle\xrightarrow{V_{i+1}}\cdots  
    \xrightarrow{U_n} \langle X_n\rangle \underbrace{\xrightarrow{V_n}}_{W_n} m_n = s_{n+1}
\]

The process of drawing wires is depicted in Figure \ref{fig:wire_overall} and consists of two stages.
\begin{enumerate}[label=(\alph*)]
\item For each $i$, draw $U_i$ and $V_i$ without crossing the previously-drawn wires so that both $U_i$ and $V_i$ have length less than $\ell_i/2$.
\item Adjust the length of $W_i$ so that its length is exactly $\ell_i$.
\end{enumerate}

We first explain how to do (b). For each $i$, reserve the space of height $40n$ on the top and bottom of the rectangle $X_i$. This space should not be used by any wires during stage (a). Thus, we have the space of area greater than $ (40n)(w_X/2) = 20nw_X$ on each side for any $i$. Then, for each incoming wire $U_i$ and $V_i$, adjust the wire $U_i$ to coil around the top of $X_i$ and the wire $V_i$ to coil around the bottom of $X_i$ as shown in Figure \ref{fig:adjust_area}. This way, we can add any even number of lengths of $U_i$ and $V_i$ not exceeding half the area, which is greater than $\frac{20nw_X}{2} > 6nw_X \geq  \ell_i/2$. Since the parity of the length of the wire is solely determined by the start and the ending point, one can get the wires to have lengths exactly $\ell_i$.

\begin{figure}[htp]
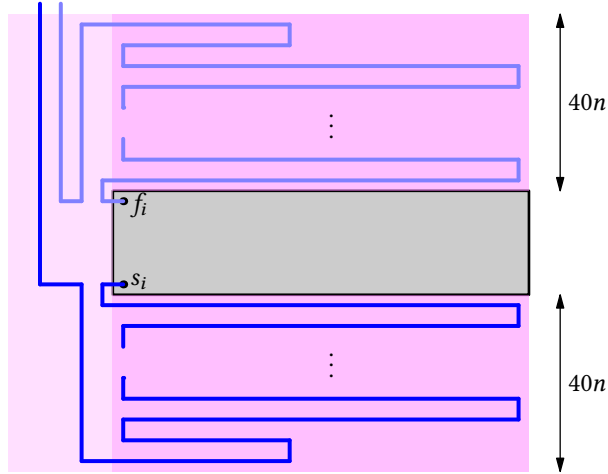

    \centering
    \asyinclude{asy/reserved_space_1}
    \caption{Packing of wires in the reserved spaces on the top and bottom of the rectangle.}
    \label{fig:adjust_area}
\end{figure}

The rest of the proof is concentrated on (a). Reserve more space of width $10n$ and height $h+80n$ on the left of each rectangle $X_i$ and previously reserved space.

Consider a broken segment $L$ along the edges in the grid formed by vertical and horizontal segments satisfying the properties that each vertical segment of $L$ is the left edges of the reserved space of some rectangles, extended by $g'/2$ on each end, and horizontal segments are allowed only at the middle of a gap between two consecutive rectangles to join the ends of two neighboring vertical segments, as shown in Figure \ref{fig:wire_overall}. 
Since $f_0$ and $s_{n+1}$ are at least $10n$ squares horizontally away from the left edges of all rectangles, $m_i$ is on the left side of all $X_i$ rectangles and $L$ for all $i$.

\begin{figure}[htp]
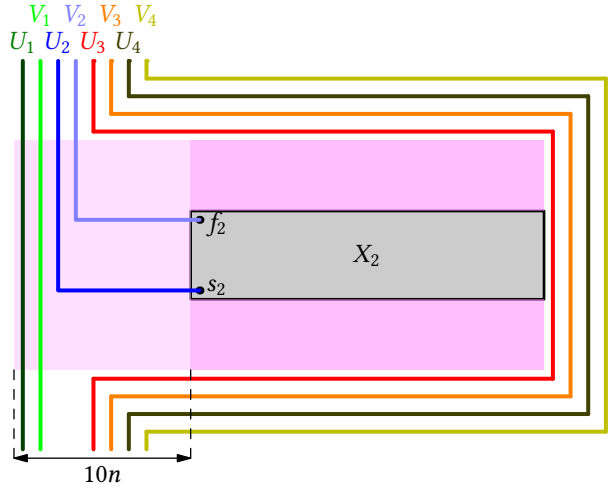

    \centering
    \asyinclude{asy/reserved_left.asy}
    \caption{Example of wires after Adjust Crossing step when $X_2$ is placed above $X_1$, $X_3$, and $X_4$. Since wires $U_1$ and $V_1$ are placed before $U_2$ and $V_2$, they do not have to detour around $X_2$.
    However, since wires $U_3$, $V_3$, $U_4$, and $V_4$ are placed after $U_2$ and $V_2$, they have to detour around $X_2$ (The lengths depicted are very not to scale.)}
    \label{fig:wire_at_block}
\end{figure}

\begin{figure}[htp]
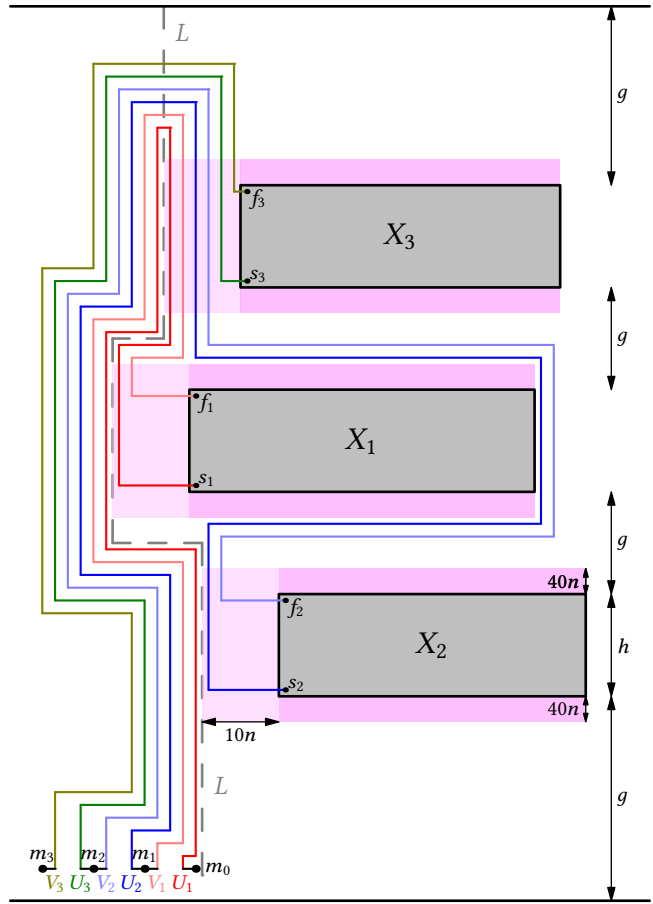

    \centering
    \asyinclude{asy/wiring_overall.asy}
    \caption{Example of wire packing when $n=3$ without length adjustment. (The lengths depicted are very not to scale.)}
    \label{fig:wire_overall}
\end{figure}

Now, we will explain how to place $U_i$ from $m_{i-1}$ to $s_i$ and $V_i$ from $m_i$ to $f_i$. Call all squares in the reserved spaces and rectangles occupied.
The process is in three steps.
\begin{enumerate}
    \item \textbf{Travel up.} We draw the wires $U_i$ and $V_i$ from $m_{i-1}$ and $m_i$, respectively to the topmost gap so that they are as close to $L$ as possible but still leaving $1$ square gap from the previous wires.  Since there are $2n$ wires, this ensures that all wires are at most $4n$ squares from all the occupied space. Thus, including all the reserved space, the wire will be at most $4n + 40n < g'/2$ away from the edges of rectangles with the condition that this part of the wires will not overlap with the rectangles or other wires. 
    \item \textbf{Drop down.} We drop these wires down on the right of $L$ to the left reserve space of $X_i$ with $1$ square gap from the previous wires (or from $L$ if they are the first pair of wires). Now, we can trace $U_i$ to $s_i$ and $V_i$ to $f_i$. During the drop, these wires may cross existing wires $U_j$ and $V_j$ for some $j<i$. Since there are $2n$ wires, the reserved space with the width $10n>(2n)(2)$ is enough to make sure that the wire will not overlap the rectangle, and that the remaining space is connected.
    \item \textbf{Adjust crossing.} Wires $U_i$ and $V_i$ may intersect wires $U_j$ and $V_j$, which are connected to rectangle $X_j$, for some $j<i$ when $U_i$ and $V_i$ attempt to go through the reserved space on the left of $X_j$. If this happens, wrap $U_i$ and $V_i$ precisely around the reserve space of $X_j$, which means that we will leave $1$ square gap from any rectangles $X_j$ or wires $U_k$ and $V_k$ ($k<i$) that previously wrapped around $X_j$. This is always possible because there are at most $2n$ wires that will wrap each rectangle $X_i$. As a result, the width of the extension $4n < 10n = g'/2 - 40n$ outside their reserved space is sufficient to contain $2n$ wires stacked on each edge of $X_{i}$ with one square gap between different wires. Note that we do not allow wire $U_i$ or $V_i$ to wrap around rectangle $X_j$ for $j > i$.
\end{enumerate}

Now, we calculate the length of wires for only one of $U_i$ or $V_i$ needed before adjusting the length by considering the horizontal length and vertical length separately and then combining. 

The vertical distance that each wire travels is at most $2H'$ because each wire $U_i$ travels directly up first and then travels down to the position without going up again, and the wire travels each direction (up or down) for at most $H'$ in total.

For the total horizontal distance of $U_i$ and $V_i$, wires travel horizontally at most once between each gap in each of traveling up and dropping down processes, so wires travel horizontally for at most $n+1$ times each process. Each time, a wire travels for the distance at most $w_X/2$. For adjusting crossing, each wire goes left and right once along the edge of those rectangles. However, there are at most $n-1$ rectangles to detour, contributing to at most $2(n-1)w_X$ squares in length. Therefore, the total distance each wire travels is at most $2(n+1)(w_X/2) + 2(n-1)w_X +2H' < 4nw_X$, so the total distance of the wire $W_i$, which combines both $V_i$ and $U_{i+1}$, is at most $8nw_X$. 
\end{proof}

\subsection{Filling the Space}
The following lemma gives a simple condition that guarantees a Hamiltonian cycle in a grid graph, which will be useful in densely filling the space in Sections \ref{sec:2D_wildcard} and \ref{sec:2xmxn_fill}. The reference of this lemma is from \cite{hampathfilling}, but we give a full statement and proof for completeness.
\begin{lemma}[Hamiltonicity of $2 \times 2$ Polygrid]
    \label{lem:hamiltonian_2x2}
    Let $T$ be a connected set of squares in a square grid. Let $T'$ be the set $T$ when each square in $T$ is replaced by a $2\times 2$ subgrid. Then, there is a Hamiltonian cycle visiting every square in $T'$.

    Moreover, the Hamiltonian cycle has the following property: for each $2\times 2$ subgrid obtained from a single square in $T$ and for each side on the boundary, the two squares on that side are adjacent in a path.
\end{lemma}
To illustrate the ``moreover'' statement, Figure \ref{subfig:2x2_induct_condition} shows the edges that are required by the lemma statement.
\begin{figure}[htp]
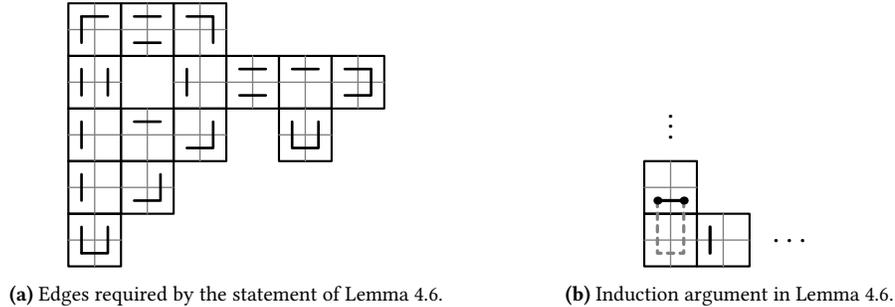

    \centering
    \begin{subfigure}{0.4\textwidth}
        \centering
        \asyinclude{asy/2x2_induct_condition.asy}
        \caption{Edges required by the statement of Lemma \ref{lem:hamiltonian_2x2}.}
        \label{subfig:2x2_induct_condition}
    \end{subfigure}
    \begin{subfigure}{0.4\textwidth}
        \centering
        \asyinclude{asy/2x2_induction.asy}
        \caption{Induction argument in Lemma \ref{lem:hamiltonian_2x2}.}
        \label{subfig:2x2_induction}
    \end{subfigure}
    \caption{The lemma for filling the space}
\end{figure}
\begin{proof}
    We do induction on $\lvert T\rvert$. The base case $\lvert T\rvert = 1$ is clear. 
    For the inductive step, let $s$ be a square in $T$ such that $T\setminus\{s\}$ is still connected. By induction hypothesis, take a Hamiltonian cycle satisfying the requirement of the lemma. Then, since $s$ is adjacent to another square $s'\in T$ and by induction hypothesis, that side of $s'$ must have the required edge. Thus, one can replace that required edge with a detour to the four squares in the subgrid of $s$ as in Figure \ref{subfig:2x2_induction}. This completes the proof.
\end{proof}

\section{Snake Cube Puzzles in \texorpdfstring{$\boldsymbol{1\times H\times W}$}{1xHxW} box}
\label{sec:2D_wildcard}
In this section, we consider the 2-dimensional variants of Snake Cube. We first consider \textsc{2D Snake Cube with Wildcards}, where we allow the wildcard \texttt * which could be used as either $\texttt S$ or $\texttt T$. We will prove the following:
\begin{theorem}
    \label{thm:2D_wildcard}
    \textsc{2D Snake Cube with Wildcards} is NP-hard.
\end{theorem}
Section \ref{subsec:wildcard_setup}, \ref{subsec:2D_fill_to_3DM}, and \ref{subsec:3DM_to_2D_fill} will prove Theorem \ref{thm:2D_wildcard}.
Then, in Section \ref{subsec:2D_pack}, we will also outline an alternative proof of the following, which was first proved in \cite{original2Dsnake}.
\begin{theorem}
    \label{thm:2D_pack}
    \textsc{2D Snake Cube Packing} is NP-hard.
\end{theorem}
\subsection{Setup}
\label{subsec:wildcard_setup}
Given an instance of Numerical 3DM $a_1,\dots,a_n$, $b_1,\dots,b_n$, $c_1,\dots,c_n$ with target sum $t$, and assuming the condition of Proposition \ref{prop:3DM_tweak} that $a_i\in (0.5t, 0.6t)$, $b_i\in (0.25t, 0.3t)$, and $c_i\in (0.125t, 0.15t)$ for all $i$, we set the following parameters to construct the input program of \textsc{2D Snake Cube with Wildcards}.
$$\setlength{\arraycolsep}{2pt}
\begin{array}{rclrcl}
g =& \text{gap width} &= 200n
\hspace{2cm} &
 H =& \text{height of the grid} &= nh + (n+1)g \\[2pt]
h =& \text{height of blocks} &= 20000n^2
& W =& \text{width of the grid} &= mt + 4g \\[2pt]
m =& \text{multiplier of widths} &= 30000n^3
& \ell_A =& \text{length of the  wires} &= 16nW \\[2pt]
 \ell_B =& \text{length of the  wires} &= 8nW 
& \ell_C =& \text{length of the  wires} &= 4nW
\end{array}$$

We will construct block gadgets $A_i$, $B_i$, and $C_i$ for all $1\leq i\leq n$. These blocks will be expected to fold into rectangles of size $h\times ma_i$, $h\times mb_i$, or $h\times mc_i$.
\begin{align*}
    \langle A_i\rangle &= 
    (\texttt{S})^{ma_i-1}
    (\texttt{TT}(\texttt S)^{ma_i-2})^{h-1}
    \texttt S\\
    \langle B_i\rangle  &= 
    (\texttt{S})^{mb_i-1}
    (\texttt{TT}(\texttt S)^{mb_i-2})^{h-1}
    \texttt S\\
    \langle C_i\rangle &= 
    (\texttt{S})^{mc_i-1}
    (\texttt{TT}(\texttt S)^{mc_i-2})^{h-1}
    \texttt S
\end{align*}

The program we will use is the following:
\begin{align*}
\mathcal P = &\langle A_1\rangle(\texttt *)^{\ell_A}\langle A_2\rangle(\texttt *)^{\ell_A}\dots
(\texttt *)^{\ell_A}\langle A_n\rangle(\texttt *)^{\ell_A} \\
&\quad \langle B_1\rangle(\texttt *)^{\ell_B}\langle B_2\rangle(\texttt *)^{\ell_B}\dots
(\texttt *)^{\ell_B}\langle B_n\rangle(\texttt *)^{\ell_B} \\
&\quad \langle C_1\rangle(\texttt *)^{\ell_C}\langle C_2\rangle(\texttt *)^{\ell_C}\dots
(\texttt *)^{\ell_C}\langle C_n\rangle(\texttt *)^{\ell_C}\langle\texttt{*}\rangle^\ell
\end{align*}
where the number of \texttt{*} at the end is to make the length of the whole string is exactly $WH$.

\subsection{Chain implies Matching}
\label{subsec:2D_fill_to_3DM}
Suppose that there is a chain that satisfies the program $\mathcal P$ above. Each block gadget forced cubes to form $h$ segments; all of these must be horizontal because otherwise, the length is too large: \[H = nh + (n+1)g = n(20000n^2) + (n+1)(200n) < 30000n^3 = m.\] It can be checked that
\begin{itemize}
    \item $H = nh + (n+1)g = nh + (n+1)(200n) < nh + 500n^2 = nh + \frac h{40}$ and
    \item $W = mt + 4g = mt + 800n < mt + 30000n^3 < m(t+1)$.
\end{itemize}
Moreover, straight paths in the same block must be in $h$ consecutive horizontal rows because the end of one segment must be next to the beginning of the other.
Therefore, by Lemma~\ref{lem:segment_packing}, there exists a matching for $(a_i)_{i=1}^n$, $(b_i)_{i=1}^n$, and $(c_i)_{i=1}^n$.
\subsection{Matching implies Chain}
\label{subsec:3DM_to_2D_fill}
Now we show that if there exists a matching in the instance of \textsc{Numerical 3DM}, i.e., permutations $\sigma$ and $\pi$ such that $a_i + b_{\sigma(i)} + c_{\pi(i)} = t$ for all $i$, then we can construct a chain that satisfies the program.

First, we note that a chain occupying a $h\times ma_i$ block can be made to satisfy the block gadget $\langle A_i\rangle$, as shown in Figure~\ref{subfig:block}. Represent $\langle A_i\rangle$ with that chain, and similarly, represent $\langle B_i\rangle$ and $\langle C_i\rangle$ with $h\times mb_i$ block and $h\times mc_i$ block.
We then place those blocks similar to Figure~\ref{subfig:reduction_outline}: divide the box into $n$ parts horizontally and putting block $\langle A_i\rangle$, $\langle B_{\sigma(i)}\rangle$, and $\langle C_{\pi(i)}\rangle$ on the $i$-th part. We also leave gaps exactly $g$ between any two adjacent blocks. Notice that since $a_i + b_{\sigma(i)} + c_{\pi(i)} = t$ for all $i$, the block placement allows leaving every gap exactly $g$. 

Then, we will apply the wire lemma to turn a string of $\texttt *$'s into a subchain connecting all these blocks. Divide the whole box into $2\times 2$ grid of cubes, and consider each $2\times 2$ grid as a $1\times 1$ square. The gaps between blocks have width of $g/2 = 100n$ squares, and the width of the grid equals $w_A = W/2$ squares for wires among $(\langle A_i\rangle)_{i=1}^n$.

For wires among $A_i$ blocks, Lemma~\ref{lem:wire} implies that there exists a sequence of $\ell_A = 8nw_A = 4nW$ squares of size $2\times 2$ that connects $A_i$ and $A_{i+1}$ for each $i$. Moreover, a subchain of cubes following a program of $16nW$ $\texttt *$ can fill the space regardless of the direction of how squares of size $2\times 2$ connect by filling each $2\times 2$ subgrid as shown in Figure~\ref{fig:fill-subgrid-wildcard} for straight squares, Figure~\ref{fig:fill-subgrid-wildcard-1}, and Figure~\ref{fig:fill-subgrid-wildcard-2} for both directions of a turn. 

For $B_i$ and $C_i$ blocks, we can set $w_X$ in the lemma to be $w_B = w_A/2$ and $w_C = w_A/4$, respectively due to the constraints $a_i\in (0.5t, 0.6t)$, $b_i\in (0.25t, 0.3t)$, and $c_i\in (0.125t, 0.15t)$ for all $i$. Thus, the wires between $B_i$ and $C_i$ blocks can be set to be $\ell_B = \ell_A/2$ and $\ell_C = \ell_A/4$, respectively, and be placed in the same way as wires between $A_i$ do.

For the wire between $A_n$ and $B_1$, they should be long enough to connect 
\begin{itemize}
    \item The last cubes of $A_n$ to the square $m_n$ for $A_i$ blocks in the notation of Lemma \ref{lem:wire}. This accounts for $\ell_A/2$ cubes.
    \item The square $m_n$ for $A_i$ blocks to the square $m_0$ for $B_i$ blocks. This accounts for at most $4W$ cubes.
    \item The square $m_0$ for $B_i$ blocks to the first tile in the $B_1$ block. This accounts for $\ell_B/2$ cubes.
\end{itemize} 
This uses $\ell_A/2+\ell_B/2+4W < \ell_A$ cubes. The extra length of cubes the space may be adjusted at the bottom space below the reserved space of $A_i$ blocks, whose area is at least $(g/2)(W/2) = gW/4 > \ell_A$ cubes. Thus, it is possible to do so.
Also, the wire between $B_n$ and $C_1$ can be placed similarly. This concludes that all blocks can be connected by wires, and there is a space of one $2\times 2$ square between wires in arbitrary arrangement.

\begin{figure}[htp]
    \centering
    \begin{subfigure}{0.25\textwidth}
        \centering
        \asyinclude{asy/subgrid_wildcard_s.asy}
        \caption{Filling ``straight" squares with wildcard}
        \label{fig:fill-subgrid-wildcard}
    \end{subfigure}
    \hspace{0.2cm}
    \begin{subfigure}{0.25\textwidth}
    \centering
    \asyinclude{asy/subgrid_wildcard_t1.asy}
        \caption{Filling ``turn" squares with wildcard, starting from the inner corner}
        \label{fig:fill-subgrid-wildcard-1}
    \end{subfigure}
    \hspace{0.2cm}
    \begin{subfigure}{0.25\textwidth}
    \centering
        \asyinclude{asy/subgrid_wildcard_t2.asy}
        \caption{Filling ``turn" squares with wildcard, starting from the outer corner}
        \label{fig:fill-subgrid-wildcard-2}
    \end{subfigure}
    \caption{Filling a $2\times 2$ square with \texttt{*} cubes.}
    \label{fig:fill-subgrid-wildcard-all}
\end{figure}

Lastly, we need to justify that the remaining space after placing all the wires and blocks can be completely filled. Notice that the construction we described so far is aligned with the $2\times 2$ polygrid, and it is connected because the blocks and wires are topologically equivalent to a path. Each piece is at least one cell apart from the others, so it cannot form a closed loop anywhere. By Lemma \ref{lem:hamiltonian_2x2}, there exists a Hamiltonian cycle that passes all the remaining squares. Thus, a string of only $\texttt{*}$ with length equal to the number of empty squares can be placed along the cycle and thus covers all the remaining space.

\subsection{Proof for Packing}
\label{subsec:2D_pack}
In this section, we will give an outline of a proof of Theorem \ref{thm:2D_pack}. The idea is mostly the same as \ref{thm:2D_wildcard}, but not allowing wildcards makes the proof different in a few places.

\paragraph{Setup.} Given an instance of Numerical 3DM $a_1,\dots,a_n$, $b_1,\dots,b_n$, $c_1,\dots,c_n$ with target sum $t$, we set up the exact same parameters $g$, $h$, $H$, $W$, $m$ as in Section \ref{subsec:wildcard_setup}. The parameters $\ell_A$, $\ell_B$, and $\ell_C$ are changed to the following to address the modulo $4$ issue, which is discussed at the end of this subsection:
\[\ell_A = 16nW + 2, \qquad \ell_B = 8nW + 2, \quad\text{and}\quad \ell_C = 4nW+2.\]
The block gadgets are also exactly the same. However, since we are not allowed to use $\texttt *$, we need to change our program to
\begin{align*}
\mathcal P ={} &\langle A_1\rangle(\texttt T)^{\ell_A}\langle A_2\rangle(\texttt T)^{\ell_A}\dots
(\texttt T)^{\ell_A}\langle A_n\rangle(\texttt T)^{\ell_A} \\
&\langle B_1\rangle(\texttt T)^{\ell_B}\langle B_2\rangle(\texttt T)^{\ell_B}\dots
(\texttt T)^{\ell_B}\langle B_n\rangle(\texttt T)^{\ell_B} \\
&\langle C_1\rangle(\texttt T)^{\ell_C}\langle C_2\rangle(\texttt T)^{\ell_C}\dots
(\texttt T)^{\ell_C}\langle C_n\rangle(\texttt T)^{\ell_C},
\end{align*}
where all $\texttt *$ are changed to $\texttt T$ and the $\langle \texttt *\rangle^\ell$ at the end is removed.

\paragraph{Chain $\boldsymbol\Rightarrow$ Matching.} This part is exactly the same as Section \ref{subsec:2D_fill_to_3DM} since it only uses block gadget.

\paragraph{Matching $\boldsymbol\Rightarrow$ Path Chain.} 
If there is a solution to an instance of Numerical 3DM, we will construct a chain satisfying $\mathcal P$. To that end, we use the exact same setup as Section \ref{subsec:3DM_to_2D_fill} by arranging the blocks according to the solution of Numerical 3DM as in Figure \ref{subfig:reduction_outline}. Like Section \ref{subsec:3DM_to_2D_fill}, we divide the grid into $2\times 2$ subgrids and ensure that all blocks fit within the subgrid. However, the major difference here is that we don't have the wildcard, so we need to wire every block using only ``Turn'' cubes. To fix this, we use the patterns in Figure \ref{fig:fill-pack} to fill it. Observe that only using patterns in Figure \ref{fig:fill-pack-t1} and Figure \ref{fig:fill-pack-t2} will make the average number of cubes per square strictly less than $4$, so the total number of cubes in a wire would be less than the specified wire length. Thus, we need to change patterns used at some turns from Figures \ref{fig:fill-pack-t1} and \ref{fig:fill-pack-t2} to \ref{fig:fill-pack-t3} and \ref{fig:fill-pack-t4}, respectively. This can adjust the number of cubes by $4$ at a time. Moreover, observe that using only patterns in Figure \ref{fig:fill-pack-t3} and Figure \ref{fig:fill-pack-t4} makes the average number of cubes per square strictly greater than $4$. This guarantees that the number of cubes used can be adjusted to be off by less than $4$ cubes from the exact length. 
\begin{figure}[htp]
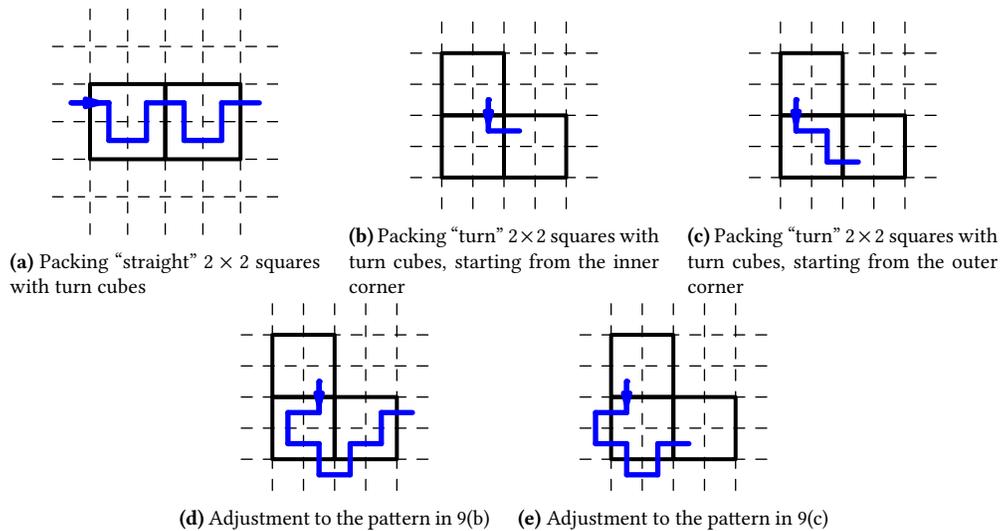

    \centering
    \begin{subfigure}{0.25\textwidth}
        \centering
        \asyinclude{asy/subgrid_pack_s.asy}
        \caption{Packing ``straight" $2\times 2$ squares with turn cubes}
        \label{fig:fill-pack-s}
    \end{subfigure}
    \hspace{0.2cm}
    \begin{subfigure}{0.25\textwidth}
    \centering
    \asyinclude{asy/subgrid_pack_t1.asy}
        \caption{Packing ``turn" $2\times 2$ squares with turn cubes, starting from the inner corner}
        \label{fig:fill-pack-t1}
    \end{subfigure}
    \hspace{0.2cm}
    \begin{subfigure}{0.25\textwidth}
    \centering
        \asyinclude{asy/subgrid_pack_t2.asy}
        \caption{Packing ``turn" $2\times 2$ squares with turn cubes, starting from the outer corner}
        \label{fig:fill-pack-t2}
    \end{subfigure}
    \\
    \begin{subfigure}{0.25\textwidth}
    \centering
        \asyinclude{asy/subgrid_pack_t3.asy}
        \caption{Adjustment to the pattern in \ref{fig:fill-pack-t1}}
        \label{fig:fill-pack-t3}
    \end{subfigure}
    \hspace{0.2cm}
    \begin{subfigure}{0.25\textwidth}
    \centering
        \asyinclude{asy/subgrid_pack_t4.asy}
        \caption{Adjustment to the pattern in \ref{fig:fill-pack-t2}}
        \label{fig:fill-pack-t4}
    \end{subfigure}
    \caption{Packing cubes in the wire aligned in $2\times 2$ subgrids}
    \label{fig:fill-pack}
\end{figure}

Finally, we need to argue that length adjustment by a multiple of $4$ (as in Figure \ref{fig:fill-pack}) is sufficient to get the exact length. To do that, we color each cube $(x,y)$ by the residues $(x\bmod 2, y\bmod 2)$. It can be shown that if the path uses the ``Turn'' instruction only, it will cycle through the four colors in a periodic fashion, making the number of cubes used invariant modulo $4$. Specifically, this number must be $2$ modulo $4$ because the width of the block is always even, forcing the wires to connect two cubes (namely, the cubes at the bottom left corner of one block and the top left corner of the other) with the same parity of $x$-coordinate but different parity of $y$-coordinate, and it is not hard to see that there exists a wire with length of $2$ modulo $4$ that follows ``Turn'' instruction only and connects two such cubes (ignoring crossing). Hence, all the wires with this property must have length $2$ modulo $4$, and length adjustments by multiples of $4$ are enough to adjust the length of the wire to any target.

Note that we did not require filling the remaining space.

\section{Snake Cube Puzzles in \texorpdfstring{$\boldsymbol{2\times H\times W}$}{2xHxW} box}
\label{sec:2xmxn_fill}

In this section, $2 \times H \times W$ \textsc{Snake Cube} is NP-complete with a reduction from Numerical 3-dimensional matching. 

\begin{theorem}
    \label{thm:2xmxn}
    $2 \times H \times W$ \textsc{Snake Cube} is NP-hard, and therefore, NP-complete.
\end{theorem}

The rest of this section is devoted to proving Theorem \ref{thm:2xmxn}. 

\subsection{Reduction Overview}

We follow the block and wire reduction infrastructure introduced in Section \ref{sec:overview_3DM} and used in Section \ref{sec:2D_wildcard}. 

Given an instance of Numerical 3DM $a_1,\dots,a_n$, $b_1,\dots,b_n$, $c_1,\dots,c_n$ with target sum $t$, and assuming the condition of Proposition \ref{prop:3DM_tweak} that $a_i\in (0.5t, 0.6t)$, $b_i\in (0.25t, 0.3t)$, and $c_i\in (0.125t, 0.15t)$ for all $i$, we set the following parameters to construct the input program of \textsc{$2\times H\times W$ Snake Cube}.

$$\setlength{\arraycolsep}{2pt}
\begin{array}{rclrcl}
g =& \text{gap width} &= 1200n,
\hspace{2cm} &
 H =& \text{height of the box} &= n(h + 6g + 4) + 2, \\[2pt]
h =& \text{height of blocks} &= 60000n^2,
& W =& \text{width of the box} &= 4g + mt + 6, \\[2pt]
m =& \text{multiplier of widths} &= 70000n^3,
& \ell_A =& \text{length of the  wires} &= 96nW,\\[2pt]
 \ell_B =& \text{length of the  wires} &= 48nW,
& \ell_C =& \text{length of the  wires} &= 24nW.
\end{array}$$

In three dimensions, each block gadgets $\langle A_i\rangle $, $\langle B_i\rangle$, and $\langle C_i\rangle$ will correspond to block of size $2 \times h \times ma_i$, $2 \times h \times mb_i$, and $2 \times h \times mc_i$. A second layer in the third dimension provides us the freedom to fill (not just pack) the space using a long sequence of \texttt T's, which we were unable to accomplish in 2D. However, the second layer also brings a major obstacle to constructing block gadgets. The naive block gadgets similar to Section \ref{sec:2D_wildcard} could inadvertently swap between layers anytime, causing an unequal distribution of rows corresponding $a_i$, $b_i$, and $c_i$ between layers that causes Lemma \ref{lem:segment_packing} to fail. To get around this issue, we will introduce a new gadget called a \vocab{shelf}. A shelf is a well-chosen set of instructions that divides the box into $n$ regions, each of height roughly $h$. The shelf will effectively constrain how and where these blocks fold so that each block will have height approximately $h$ in each layer. The exact description will be provided in Section \ref{subsec:shelf_cons}.

Here, we provide an overview of the construction. We start our program with a sequence of instructions, denoted $\langle\text{shelf}\rangle$, that forces the construction of the \emph{shelf}, which will be provided in Section \ref{subsec:shelf_cons}. As described in Section \ref{sec:overview_3DM}, the core of the program will consist of block gadgets
\begin{align*}
    \langle A_i\rangle &= 
    (\texttt{S})^{ma_i-1}
    (\texttt{TT}(\texttt S)^{ma_i-2})^{2h-1}
    \texttt S,\\
    \langle B_i\rangle  &= 
    (\texttt{S})^{mb_i-1}
    (\texttt{TT}(\texttt S)^{mb_i-2})^{2h-1}
    \texttt S,\\
    \langle C_i\rangle &= 
    (\texttt{S})^{mc_i-1}
    (\texttt{TT}(\texttt S)^{mc_i-2})^{2h-1}
    \texttt S.
\end{align*}
Block $\langle A_i\rangle$, $\langle B_i\rangle$, and $\langle C_i\rangle$ are expected to be represented by a chain that occupies blocks of size $2\times h\times ma_i$, $2\times h\times mb_i$, and $2\times h\times mc_i$ respectively, as shown in Figure \ref{fig:2xmxn_block}.
\begin{figure}[htp]
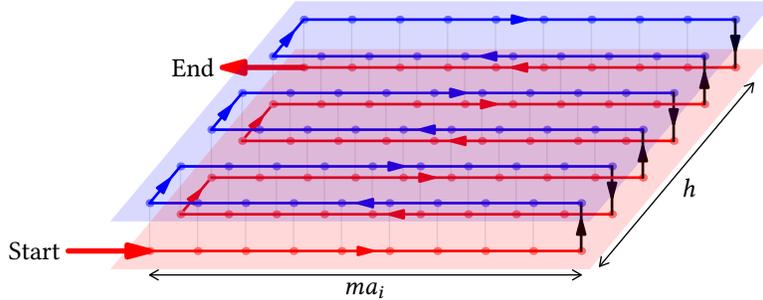

    \centering
    \asyinclude{asy/2xmxn_block.asy}
    \caption{Illustration of chain satisfying block gadget $\langle A_i\rangle$ in $2\times H\times W$ \textsc{Snake Cube}. The colors red and blue represent different layers of cubes.}
    \label{fig:2xmxn_block}
\end{figure}

The wiring gadget consists of a long sequence of $\texttt T$'s. Lastly, we end with a long sequence of turns $(\texttt T)^\ell$ that provides just enough cubes to fill the remaining space if the instance \textsc{Numerical 3DM} is solvable. 

The program is then the following.
\begin{align*}
\mathcal P = \langle\text{shelf}\rangle&\langle A_1\rangle(\texttt T)^{\ell_A}\langle A_2\rangle(\texttt T)^{\ell_A}\dots
(\texttt T)^{\ell_A}\langle A_n\rangle(\texttt T)^{\ell_A} \\
&\quad \langle B_1\rangle(\texttt T)^{\ell_B}\langle B_2\rangle(\texttt T)^{\ell_B}\dots
(\texttt T)^{\ell_B}\langle B_n\rangle(\texttt T)^{\ell_B} \\
&\quad \langle C_1\rangle(\texttt T)^{\ell_C}\langle C_2\rangle(\texttt T)^{\ell_C}\dots
(\texttt T)^{\ell_C}\langle C_n\rangle(\texttt T)^{\ell_C}(\texttt T)^\ell,
\end{align*}
where $\ell$ is such that the entire length of the program is $2HW$.

\subsection{``Shelf'' Construction}
\label{subsec:shelf_cons}

We begin by specifying the set of instructions that will force the formulation of a shelf with $n$ rows, shown below.

\begin{figure}[htp]
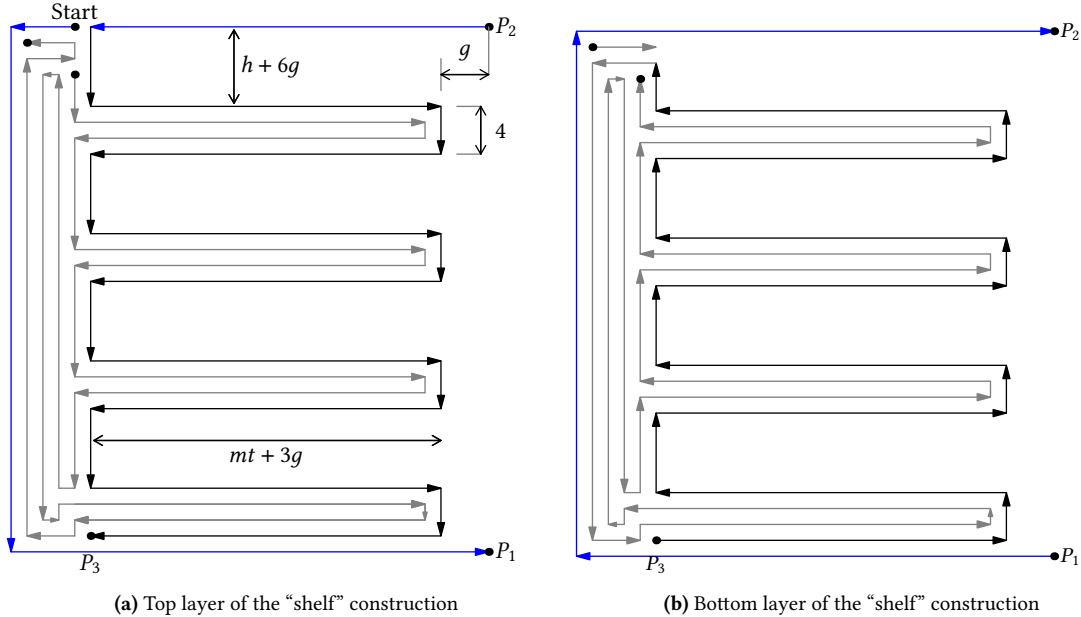

\centering
\begin{subfigure}{0.45\textwidth}
\asyinclude{asy/shelf_top.asy}
\caption{Top layer of the ``shelf'' construction}
\label{fig:topshelf}
\end{subfigure}
\begin{subfigure}{0.45\textwidth}
\asyinclude{asy/shelf_bottom.asy}
\caption{Bottom layer of the ``shelf'' construction}
\label{fig:bottomshelf}
\end{subfigure}
\caption{The ``shelf'' configuration highlighted in \ref{subsec:shelf_cons}. Whenever there is a discontinuity, the chain moves to the other layer.}
\label{fig:shelf}
\end{figure}


We first trace out the blue \emph{frame} of the shelf (shown in Figure \ref{fig:shelf}) traversing from the starting point labeled ``Start'' to $P_1$ in the top layer, $P_1$ to $P_2$ in the bottom layer, and finishing the border in the top layer. 

\[ \langle\text{outer frame}\rangle = \texttt{SSSST}\texttt S^{H-2}\texttt{T}\texttt S^{W-2}\texttt{TT}\texttt S^{W-2}\texttt{T}\texttt S^{H-2}\texttt{T}\texttt S^{W-2}\texttt{TT}\texttt S^{W-6}. \]

We now trace out the black frame $n$ shelves by following the path in black. Each shelf can be traced with the instructions 
$$\langle\text{top shelf frame}\rangle = \texttt{ T}\texttt S^{h+6g} \texttt{T}\texttt S^{W-g-6}\texttt{TSST}\texttt S^{W-g-6}.$$
We repeat this $n$ times to trace out the shelves in the top layer. At point $P_3$, we switch to the bottom layer and trace the $n$ shelves for the bottom layer as well.
The tracing instructions for the bottom layer is the reverse of the top one:
$$\langle\text{bottom shelf frame}\rangle = (\texttt S)^{W-g-6} \texttt{TSST}\texttt S^{W-g-6} \texttt{T}\texttt S^{h + 6g} \texttt{T}.$$
Then, the entire structure of the frame is
\[ 
\langle \text{shelf frame}\rangle
= \langle \text{outer frame}\rangle
\left( \langle\text{top shelf frame}\rangle \right)^{n} \texttt{TT} \left( \langle\text{bottom shelf frame}\rangle \right)^{n-1} (\texttt S)^{W-g-6} \texttt{TSST}\texttt S^{W-g-6} \texttt{T}\texttt S^{h-2},\]
where the ending instructions trace out the last shelf in the bottom layer. This completes the frame of the shelf. It is not hard to see the construction of the frame is unique. (Notably, whenever there is a continuous sequence of $\Theta(W)$ straights, the corresponding cubes must be aligned in a line along the length of the box, constraining the orientation and placement of all cubes along the frame to the unique representation.)

We now fill the inside of the shelf (shown in gray in Figure \ref{fig:shelf}). If we follow the instructions above with $\texttt{TSST}\texttt S^{H-5}$, the chain must turn into the shelf as there is no other way to place the straight line of $H-4$ cubes otherwise. Now, observe that there is only one exit to the inside of the shelf construction. If the snake cube box is to be filled, then the inside of the shelf must be completely filled before the chain of cubes exits the shelf; the exact filling instructions is not important. One possible way is to start filling the bottom layer, move up, and then fill the top layer.
\begin{align*}
&\langle{\text{bottom filling}}\rangle \\
&\quad=
\texttt{TSST}(\texttt S)^{H-5}\texttt{TSSTT}(\texttt S)^{W-g-6}\texttt{TT}(\texttt{S})^{W-g-5}\texttt{TTT}(\texttt{S})^{H-7}\texttt{TT}(\texttt  S)^{H-9}\texttt{TT}\left((\texttt S)^{h+6g+1}\texttt{T}(\texttt{S})^{W-g-6}\texttt{TT}(\texttt{S})^{W-g-6}\texttt{TS} \right)^{n-1}\texttt{S},
\\
&\langle{\text{top filling}}\rangle \\
&\quad =\left(\texttt{ ST}\texttt S^{W-g-6} \texttt{TT}\texttt S^{W-g-6}\texttt{T}\texttt S^{h+6g+1}\right)^{n-1}\texttt{TT}\texttt S^{H-9}\texttt{TT}\texttt S^{H-7}\texttt{TTT}\texttt S^{W-g-5}\texttt{TT}\texttt S^{W-g-6}\texttt{TTSST}\texttt S^{H-5} \texttt{TSSTTSS}.
\end{align*}

The above instructions can form the gray subchain shown in Figure \ref{fig:shelf}. The final construction of the shelf is then
$$\langle\text{shelf}\rangle
= \langle \text{shelf frame}\rangle \langle\text{top filling}\rangle \texttt{TT} \langle\text{bottom filling}\rangle \texttt{TTSSSS}.$$

\subsection{Chain implies Matching }

We now show that a solution to \textsc{$2\times H\times W$ Snake Cube} implies a matching for \textsc{Numerical 3DM}. This time, we do not need Lemma \ref{lem:segment_packing} because the shelf already significantly constrained the construction. We simply argue directly.

Let $X_i$ denote either $A_i$, $B_i$, or $C_i$. Since $m > H$, each of the long straight sequences of length $mx_i - 2$ in $\langle X_i \rangle$ must form a continuous segment of $mx_i$ cubes lying in the horizontal direction. Additionally, since $m > g$, all cubes within a $\langle X_i \rangle$ block gadget must lie entirely within one row of the shelf. 

As a result, all of $2h$ segments of length $mx_i$ in each block must entirely lie in the volume $2\times (h + 6g)\times (W-6)$ of one row of the shelf. We claim that each row of shelf must contain exactly one $A_i$, $B_i$, and $C_i$ block. To see this, observe that it is impossible to have two $A_i$ blocks in the same row of the shelf; otherwise, if $A_i$ and $A_j$ lie in the same row of the shelf, there are $2\cdot 2h$ segments of length at least $m\min\{a_i,a_j\} \geq 0.5m(t+1) > W/2$. Since $4h > 2(h+6g)$, there must be two segments in the same row. Their combined length is at least $2m\min\{a_i,a_j\} > W$, yielding a contradiction. Thus, each row must contain exactly one $A_i$ block. Assuming this result, we can also show that each row must contain exactly one $B_i$ block and one $C_i$ block as desired.

Now consider each $X_i$ block as $2h$ straight segment of size $1\times mx_i$. We claim that in a row of the shelf that contains block $A_i$, $B_j$, and $C_k$ there exists a $1\times 1\times (W-6)$ row that contains three segments of lengths $ma_i$, $mb_j$, and $mc_k$. To prove that this row exists, we can use a similar argument to Lemma \ref{lem:segment_packing}. Define $n_A, n_B,n_C, n_{\text{good}}$ as same as in the proof of Lemma \ref{lem:segment_packing} but counting only rows in between shelf we are considering (but summing across both layers), which has size $2\times (h + 6g)\times (W-6)$. Thus, $n_A + n_B + n_C + n_{\text{good}} = 2(h+6g)$. However, by counting the number of A-segments, B-segments, and C-segments, we get, respectively, that
\begin{align*}
    n_A &= 2(h+6g) - 2h = 12g \\
    2h \leq 3n_A + n_C + n_{\text{good}} \implies n_B &\leq 2n_A + 2(h+6g) - 2h \leq 36g\\
    2h \leq 7n_A + 3n_B + n_{\text{good}} \implies n_C &\leq 6n_A + 2n_B + 2(h+6g) - 2h = 156g
\end{align*}
Since $2(h+6g) > 12g + 36g + 156g$ for all $n>1$, there must be at least one good row as desired.

Consider the row that contains all segments of length $ma_i$, $mb_j$, and $mc_k$, we must have 
\begin{align*}
    ma_i + mb_j + mc_k &\leq W = mt + 4g + 6 < m(t+1),\\
    a_i + b_j + c_k &\leq t.
\end{align*}

Summing the inequality above over all rows, we get $nt = \sum_i a_i + \sum_j b_j + \sum_k c_k \leq nt$, so therefore the equality $a_i + b_j + c_k = t$ must hold for all rows of the shelf. This implies a path filling necessarily implies the existence of matching.

\subsection{Matching implies Chain }

Now we show that if there exists a matching in \textsc{Numerical 3DM}, i.e., permutations $\sigma$ and $\pi$ such that $a_i + b_{\sigma(i)} + c_{\pi(i)} = t$ for all $i$, then we can construct a chain that satisfies the program.

First, we note that a chain occupying a $2\times h\times ma_i$ block can be made to satisfy the block gadget $\langle A_i\rangle$. Note that we will need the starting and ending points to be on different rows, which can be done by letting the chain alternate between the top and bottom layers every two rows, as shown in Figure \ref{fig:2xmxn_block}. Represent $\langle A_i\rangle$ with that chain, and similarly, represent $\langle B_i\rangle$ and $\langle C_i\rangle$ with $2\times h\times mb_i$ block and $2\times h\times mc_i$ block.
We then place those blocks similar to Figure \ref{subfig:reduction_outline}: divide the box into $n$ parts horizontally and putting block $\langle A_i\rangle$, $\langle B_{\sigma(i)}\rangle$, and $\langle C_{\pi(i)}\rangle$ on the $i$-th part. We also leave gaps of exactly $g$ between any two adjacent blocks and gaps of $3g$ between each block and the shelf. Notice that since $a_i + b_{\sigma(i)} + c_{\pi(i)} = t$ for all $i$, the block placement allows for leaving gaps of exactly $g$ between adjacent blocks. 

To connect blocks together, we observe the following property of traversing with a long sequence of turns in 3-dimensional space. 

\begin{lemma}[$2 \times 2 \times 2$ Traversal]
    Given the instructions $(\texttt T)^8$ and a $2 \times 2 \times 2$ cube, if the subchain following these instructions enters the cube from one face, it may exit the cube from any other face. 
    \label{lem:zigzag}
\end{lemma}

\begin{proof}
    This result is shown in Abel et. al.'s paper \cite{abel2013finding}. We include the images from their paper in Figure \ref{fig:abelcubecombo} to provide a visual proof of our claim above. 
    \begin{figure}[htp]
    \centering
        \includegraphics[width=2.5cm]{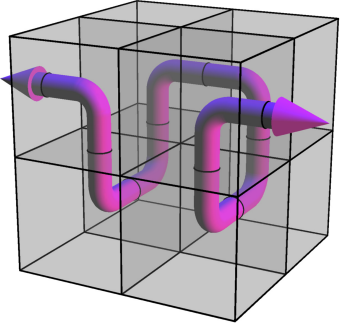}
        \hspace{0.3cm}
        \includegraphics[width=2.5cm]{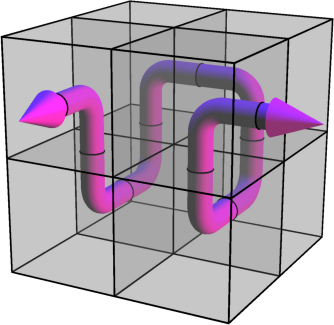}
        \hspace{0.3cm}
        \includegraphics[width=2.5cm]{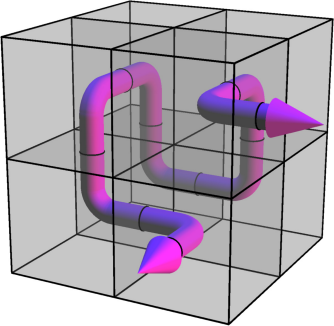}
        
        \includegraphics[width=2.5cm]{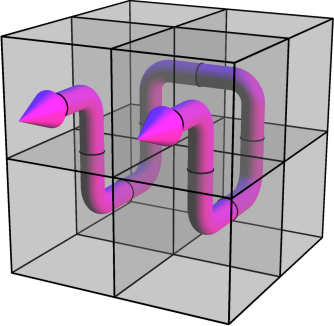}
        \hspace{0.3cm}
        \includegraphics[width=2.5cm]{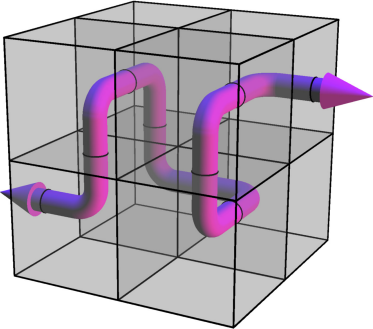}
        \hspace{0.3cm}
        \includegraphics[width=2.5cm]{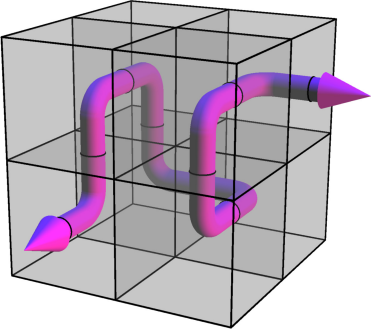}
        \caption{All possible configurations of a subchain entering from one face and exiting through another face, as shown in \cite[Figure 6,7]{abel2013finding}.}
        \label{fig:abelcubecombo}
    \end{figure}
\end{proof}

Since $h$ and $m$ are multiples of $4$, the remaining unfilled space in the box can be partitioned into $2 \times 4 \times 4$ blocks of space. From here on, we view the remaining space as a smaller 2D grid where each cell corresponds to a $2 \times 2 \times 2$ cube in the original box. The unfilled space in the box is then a connected $ 2 \times 2$ polygrid. Traversing between $2 \times 2 \times 2$ blocks of space with $(\texttt T)^8$'s gives the same movement freedom as traversing between cells in a 2D grid given a wildcard. Therefore, the remaining proof closely follows that in Section \ref{sec:2D_wildcard}. However, the shelf poses an obstacle, and wires must detour around the shelf.

We define $W^* = \frac{W}{2}$, $g^* = g/2 = 600n$, $\ell_A^* = \ell_A/8 = 24nW^*$, $\ell_B^* = \ell_B/8 = 12nW^*$, and $\ell_C^* = \ell_C/8 = 6nW^*$ which are the equivalent quantities of $W$, $g$, $\ell_A$, $\ell_B$, and $\ell_C$ in the new $1 \times 0.5H \times 0.5W$ box. By Lemma \ref{lem:wire}, it is possible to connect each of the $\langle X_i\rangle$ blocks in order while ensuring the remaining space is continuous $2 \times 2$ polygrid, if we, as in Section \ref{sec:2D_wildcard}, allow the wires to be two cells thick. Note these constants are $\frac{3}{2}$ times those in Section \ref{sec:2D_wildcard} because the wires must detour around the shelf, as mentioned briefly in Lemma \ref{lem:wire}. (To detour around the shelf, each wire will have an additional $4nw_X$ length in the notation of Lemma \ref{lem:wire} because each wire has to detour around at most $n$ shelf frame, each for $2$ times. Each time, the wire has to detour with length at most $2w_X$. This sums up to $2n(2w_X) = 4nw_X$ as desired.) 

Moreover, when all the wires detour around the same shelf, wires need to stack into more layers, requiring wider gaps. Note that there are $3n$ wires in total, and each wire contributes to at most $4$ layers around each shelf (each wire is split into $U_i$ and $V_i$, both need to travel up and drop down). We need to leave a square between each layer of wires, the stack of wires can fit in the gap of width $2(2 \cdot (3n \cdot 4)) = 48n < g^*$. (Note that a square means $2\times 2$ polygrid in the equivalent 2D grid.) This justifies that the extra gap of width $3g$ on each shelf and the gap of width $g$ on the right of the box is sufficient. 

\begin{figure}[htp]
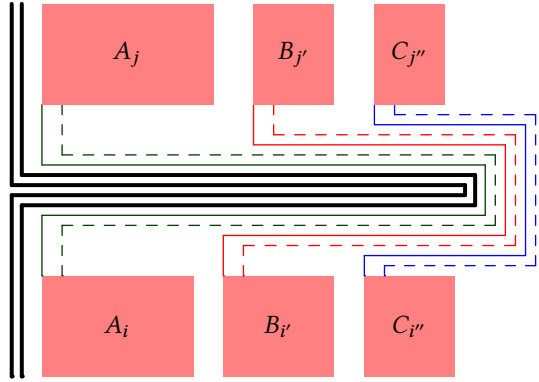

\centering
\asyinclude{asy/shelf_detour.asy}
\caption{Wires detour around the shelf}
\label{fig:detourshelf}
\end{figure}

All that is left is to fill the remaining space. Since the remaining space is $2 \times 2$ polygrid (or a $2 \times 4 \times 4$ polygrid when viewed in three dimensions), we can fill it by Lemma \ref{lem:hamiltonian_2x2}.

Therefore, there exists a matching \textsc{Numerical 3DM} if and only if there exists a solution to snake cube $2\times H\times W$ \textsc{Snake Cube}. Since our reduction takes polynomial time, $2 \times H \times W$ \textsc{Snake Cube} is NP-hard.

\section{Snake Cube Puzzles with Hexagonal Prisms}
\label{sec:triangular}
In this section, we consider a version of a 2D Snake cube with a chain of hexagonal prisms. From now on, we represent each prism by points. Two prisms are adjacent if and only if the points corresponding to them are adjacent in the triangular lattice. Thus, the problem becomes a triangular grid variant of the flattening fixed-angle chains problem in \cite{original2Dsnake}.

An infinite \vocab{triangular grid} is a two-dimensional lattice generated by vectors $v_1=\left(\begin{smallmatrix} 1\\ 0\end{smallmatrix}\right)$ and $ v_2=\left(\begin{smallmatrix} \cos 60^\circ\\ \sin 60^\circ\end{smallmatrix}\right)$; each point represents a hexagonal prism. Two points in a triangular grid are \vocab{adjacent} if they are distance $1$ apart. A \vocab{$\boldsymbol{60^\circ}$ parallelogram box} of dimension $H\times W$ is the set of $HW$ points obtained by translating the set $\{iv_1 + jv_2 : i\in\{1,\dots,W\}, j\in\{1,\dots,H\}\}$ by some lattice vector.

For this section, a \vocab{program} is a string that consists of only characters $\texttt S$, $\texttt T_{60}$, and $\texttt T_{120}$, where $\texttt S$ denotes straights, $\texttt T_{60}$ denotes $60^\circ$ turns (forming $120^\circ$ angle), and $\texttt T_{120}$ denotes $120^\circ$ turns (forming $60^\circ$ angle). We say that a chain $C = (p_1,p_2,\dots,p_{\lvert s\rvert})$ (length $\lvert s\rvert$) satisfies $s$ if and only if for every $i\in\{2,3,\dots,\lvert s\rvert-1\}$, the angle between $p_{i-1}$, $p_i$, $p_{i+1}$ is $180^\circ$ if $s_i = \texttt S$, $60^\circ$ if $s_i=\texttt T_{120}$, and $120^\circ$ if $s_i=\texttt T_{60}$. $C$ is \vocab{closed} if and only if $p_1=p_{|s|}$.

We will consider two versions of this problem.
\begin{itemize}
    \item \textsc{Triangular Closed Chain}: given an program $\mathcal P$, decide there is a closed chain that satisfies the program $\mathcal P$.
    \item \textsc{Bounded Triangular Path Packing}:  given a $60^\circ$ parallelogram box $B$, an program $\mathcal P$, and two adjacent vertices $u$ and $v$ on a boundary of $B$, decide whether there is a chain from $u$ to $v$ that satisfies program $s$.
\end{itemize}
In particular, the \textsc{Triangular Closed Chain} version is the triangular analog of the flattening problem in \cite{original2Dsnake}. We will prove the following.
\begin{theorem}
Both \textsc{Triangular Closed Chain} and \textsc{Bounded Triangular Path Packing} are NP-complete.
\end{theorem}
Clearly, these two problems are in NP with certificates being the sequence of points in the chain. The rest of this section is devoted to proving that these two problems are NP-hard. In particular, Section \ref{subsec:triangular_block}, \ref{subsec:triangular_packing}, and \ref{subsec:triangular_wiring} will show that \textsc{Bounded Triangular Path Packing} is NP-hard by using a similar reduction from Numerical 3DM following the same outline in Section \ref{sec:overview_3DM}, except that everything will be sheared to align with the parallelogram. Then, in Sections \ref{subsec:frame}, we will show how to convert \textsc{Bounded Triangular Path Packing} to \textsc{Triangular Closed Chain} by introducing a new gadget called a \vocab{frame}.
\begin{remark}
    With the same proof, we can show that the problem of ``given a $60^\circ$ parallelogram box $B$ and a program $\mathcal P$, decide whether there is a chain of points in $B$ satisfying instruction $\mathcal P$'' is NP-Complete. We omit the proof for this version.
\end{remark}
\subsection{Setup and the Block Gadget}
\label{subsec:triangular_block}
The reduction will be very similar to the one outlined in Section \ref{sec:overview_3DM}. However, we will shear everything to align with the parallelogram. In particular, when we talk about the \vocab{width} and the \vocab{height} of the grid, we mean the length parallel to the sides of $B$.

Same as before, we fix an instance of $3$-partition $a_1,\dots,a_n$, $b_1,\dots,b_n$, $c_1,\dots,c_n$ with target sum $t$. Assume the condition of Proposition \ref{prop:3DM_tweak} that $a_i\in (0.5t, 0.6t)$, $b_i\in (0.25t, 0.3t)$, and $c_i\in (0.125t, 0.15t)$ for all $i$.

We set the following parameters.
$$\setlength{\arraycolsep}{2pt}
\begin{array}{rclrcl}
g =& \text{gap width} &= 300n,
\hspace{2cm} &
 H =& \text{height of the box} &= nh + (n+1)g, \\[2pt]
h =& \text{height of blocks} &= 60000n^3,
& W =& \text{width of the grid} &= mt + 4g, \\[2pt]
m =& \text{multiplier of widths} &= 90000n^3,
& \ell_A =& \text{length of the  wires} &= 32nW/3,\\[2pt]
 \ell_B =& \text{length of the  wires} &= 16nW/3,
& \ell_C =& \text{length of the  wires} &= 8nW/3.
\end{array}$$

We take an $H\times W$ parallelogram box depicted in Figure \ref{fig:triangular_setup}. First, we have the following strings of block gadgets for all $i=1,2,\dots,n$
\begin{align*}
    \langle A_i\rangle &= 
    (\texttt S)^{ma_i-1}
    \texttt{T}_{120}\texttt{T}_{60}\,
    (\texttt S)^{ma_i-2}  \left(\texttt{T}_{60}\texttt{T}_{120}\,
    (\texttt S)^{ma_i-2}\,
    \texttt{T}_{120}\texttt{T}_{60}\,
    (\texttt S)^{ma_i-2}
    \right)^{h/2-1}\texttt S, \\
    \langle B_i\rangle &= 
    (\texttt S)^{mb_i-1}
    \texttt{T}_{120}\texttt{T}_{60}\,
    (\texttt S)^{mb_i-2}  \left(\texttt{T}_{60}\texttt{T}_{120}\,
    (\texttt S)^{mb_i-2}\,
    \texttt{T}_{120}\texttt{T}_{60}\,
    (\texttt S)^{mb_i-2}
    \right)^{h/2-1}\texttt S, \\
    \langle C_i\rangle &= 
    (\texttt S)^{mc_i-1}
    \texttt{T}_{120}\texttt{T}_{60}\,
    (\texttt S)^{mc_i-2}  \left(\texttt{T}_{60}\texttt{T}_{120}\,
    (\texttt S)^{mc_i-2}\,
    \texttt{T}_{120}\texttt{T}_{60}\,
    (\texttt S)^{mc_i-2}
    \right)^{h/2-1}\texttt S,
\end{align*}
which folds into a parallelogram block as depicted in Figure \ref{fig:triangular_setup}.

The program we will use is
\begin{align*}
\mathcal P ={} &\langle A_1\rangle\,(\texttt T_{60})^{\ell_A}\,\langle A_2\rangle\,(\texttt T_{60})^{\ell_A}\dots
(\texttt T_{60})^{\ell_A}\,\langle A_n\rangle\,(\texttt T_{60})^{\ell_A} \\
&\langle B_1\rangle\,(\texttt T_{60})^{\ell_B}\,\langle B_2\rangle\,(\texttt T_{60})^{\ell_B}\dots
(\texttt T_{60})^{\ell_B}\,\langle B_n\rangle\,(\texttt T_{60})^{\ell_B} \\
&\langle C_1\rangle\,(\texttt T_{60})^{\ell_C}\,\langle C_2\rangle\,(\texttt T_{60})^{\ell_C}\dots
(\texttt T_{60})^{\ell_C}\,\langle C_n\rangle\,(\texttt T_{60})^{\ell_C}.\\
\end{align*}%
\subsection{Chain implies Matching}
\label{subsec:triangular_packing}
Suppose that there is a chain that satisfies program $\mathcal P$ above. 
Observe that the block gadget must form $h$ straight segments of length $ma_i$, $mb_i$, or $mc_i$. All of these straight segments must be horizontal with respect to the longer sides; otherwise, placing any of the segments in $60^{\circ}$ or $120^{\circ}$ to the longer sides will make the segments exceed the boundary because
    \[H = nh + (n+1)g = n(60000n^2) + (n+1)(200n) < 90000n^3 = m.\]
    It can be checked that 
\begin{itemize}
    \item $H = nh + (n+1)g = nh + (n+1)(600n) < nh + 1500n^2 = nh + \frac h{40}$
    \item $W = mt + 4g = mt + 2400n < mt + 90000n^3 < m(t+1)$.
\end{itemize}
Moreover, straight paths in the same block must be in $h$ consecutive horizontal rows because the end of one segment must be next to the beginning of the other.
Therefore, applying Lemma \ref{lem:segment_packing}, there is a solution to Numerical 3DM for $(a_i)_{i=1}^n$, $(b_i)_{i=1}^n$, and $(c_i)_{i=1}^n$.

\begin{figure}[htp]
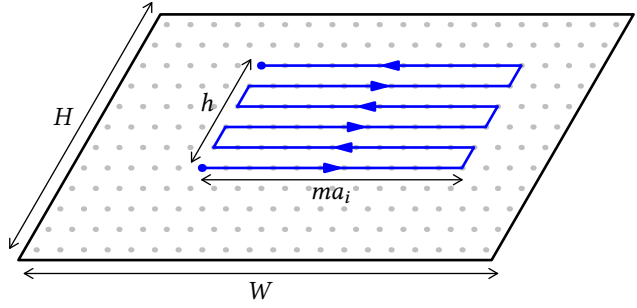

    \centering
    \asyinclude{asy/triangular_setup.asy}
    \caption{An example of the bounding box and a block gadget. The gray dots represent prisms in the triangular grid.}
    \label{fig:triangular_setup}
 \end{figure}
\subsection{Matching implies Chain}
\label{subsec:triangular_wiring}
\textsc{Numerical 3DM}, i.e., permutations $\sigma$ and $\pi$ such that $a_i + b_{\sigma(i)} + c_{\pi(i)} = t$ for all $i$, then we can construct a chain that satisfies the program.

First, we note that a chain occupying a $h\times ma_i$ parallelogram block can be made to satisfy the block gadget $\langle A_i\rangle$, shown in Figure \ref{fig:triangular_setup}. Represent $\langle A_i\rangle$ with that chain, and similarly, represent $\langle B_i\rangle$ and $\langle C_i\rangle$ with $h\times mb_i$ parallelogram block and $h\times mc_i$ parallelogram block.
We then place those blocks similar to Figure \ref{subfig:reduction_outline}, but sheared by $60^\circ$: divide the box into $n$ parts horizontally and putting block $\langle A_i\rangle$, $\langle B_{\sigma(i)}\rangle$, and $\langle C_{\pi(i)}\rangle$ on the $i$-th part. We also leave gaps exactly $g$ between any two adjacent blocks. Notice that since $a_i + b_{\sigma(i)} + c_{\pi(i)} = 3t$ for all $i$, the block placement allows leaving every gap exactly $g$. 

Now, we need to connect between blocks using a sequence of $\texttt T_{60}$'s. First, in the triangular grid, we embed the hexagonal grid as shown in Figure \ref{subfig:hex_grid}.
\begin{figure}[htp]
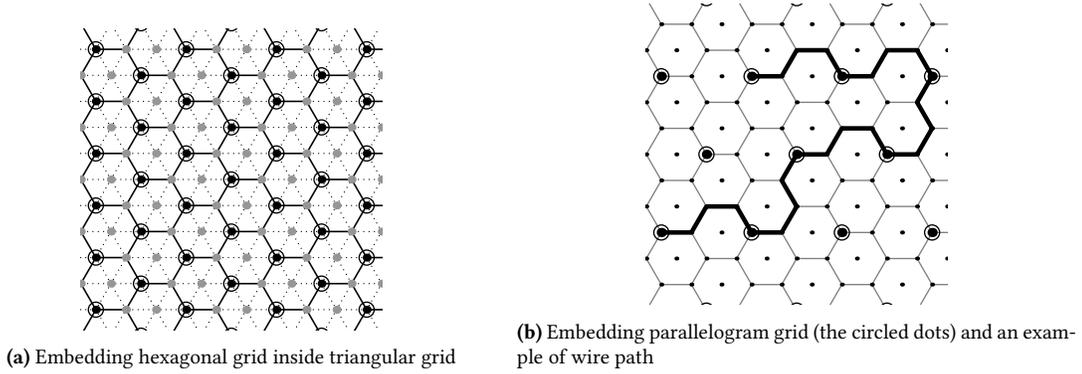

    \centering
    \begin{subfigure}{0.45\textwidth}
        \centering
        \asyinclude{asy/triangular_grid_1.asy}
        \caption{Embedding hexagonal grid inside triangular grid}
        \label{subfig:hex_grid}
    \end{subfigure}
    \begin{subfigure}{0.45\textwidth}
        \centering
        \asyinclude{asy/triangular_grid_2.asy}
        \caption{Embedding parallelogram grid (the circled dots) and an example of wire path}
        \label{subfig:wire_path}
    \end{subfigure}
    \caption{Wiring the blocks together}
\end{figure}

In the hexagonal grid (Figure \ref{subfig:hex_grid}), it is easy to see that one can move between two adjacent (distance $\sqrt 3$ unit apart) black vertices by traveling through solid edges only. Since all angles between solid edges are $120^\circ$, traveling through solid edges uses only a sequence of $\texttt T_{60}$'s. 

Next, we embed a parallelogram grid inside the grid of black vertices as shown in Figure \ref{subfig:wire_path}. This is a $3\times 3$-refinement of the original parallelogram grid. Thus, we may now apply the wire lemma (Lemma \ref{lem:wire}) on the entire figure scaled down by $3$. In particular, to connect the $\langle A_i\rangle$ blocks, we use parameters $g' = g/3 = 100n$, $w_A = W/3$, and travel between two adjacent squares correspond to four movements in the triangular grids, giving $\ell_A = 4(8nw_A) = 32nW/3$ tiles. Blocks $\langle B_i\rangle$'s and $\langle C_i\rangle$'s can be connected similarly with wire lengths divided by two and by four respectively. Finally, we transform the sequence from the embedded grid to the original triangular grid accordingly; an example of a wire path is shown in Figure \ref{subfig:wire_path}. This implies that there is the desired path.
%
%
%
%
\subsection{Reduction from Closed Chain to Bounded Path}
\label{subsec:frame}
To turn \textsc{Bounded Triangular Path Packing} to \textsc{Triangular Closed Chain}, we need a new structure called a \vocab{frame}, which is a structure that confines the rest of the construction into a bounded parallelogram, making the reduction similar to previous subsections work. This construction is similar to the frame gadget in \cite{original2Dsnake}, but we need to do a parity in modulo a much larger prime since triangular grids have more possibilities for angles of each edge.

\bigskip

The reduction goes as follows. Suppose that $\mathcal P$ is the instance of \textsc{Bounded Triangular Packing}. Then, we first blow up $\mathcal P$ by $101$ times by replacing
$$\texttt S \text{ with } \underbrace{\texttt{SSS}\dots\texttt S}_{101},
\qquad \texttt{T}_{60} \text{ with } \underbrace{\texttt{SSS}\dots\texttt S}_{100}\texttt T_{60}, 
\quad\text{ and }\quad \texttt{T}_{120} \text{ with }  \underbrace{\texttt{SSS}\dots\texttt S}_{100}\texttt T_{120}.$$
This gives a new program $101\mathcal P$ that will constrain a chain in the bounding box $101B$ and will connect two vertices $101u$ and $101v$ on the boundary of $101B$ that has distance $101$ apart. We now append this with the frame program $F$, defined to be folded as in Figure \ref{fig:triangular_frame}.
\begin{figure}[htp]
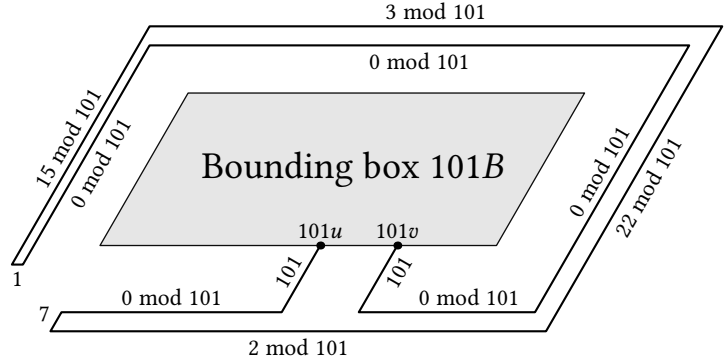

    \centering
    \asyinclude{asy/triangular_frame.asy}
    \caption{The Frame Gadget}
    \label{fig:triangular_frame}
\end{figure}

\bigskip

We now claim that the concatenation of $101\mathcal P$ and $F$ is a valid reduction to \textsc{Triangular Closed Chain}. To prove this, we first show that the frame works as expected.
\begin{claim}
    \label{claim:frame_is_forced}
    Any chain satisfying program $F$ must look like \ref{fig:triangular_frame} up to global rotation and reflection.
\end{claim}
\begin{proof}
    We first note that the vertices in each line segment must be connected by the ``straight'' instruction, so it must be a long line segment. There are six segments whose length is not $0\pmod{101}$, lengths $1$, $2$, $3$, $7$, $15$, $22$, and we need to group them into three groups (corresponding to three possible orientations of the segment) and assign a $+$ or $-$ sign for each segment so that the sum turns out to be equal.
    An exhaustive check reveals that the only way to do so in modulo $101$ is $\{-1,-2,3\}$, $\{-7,-15,22\}$, and $\emptyset$.
    This forces all segments to be at one of the two angles; without loss of generality, let them be as shown in Figure \ref{fig:triangular_frame}. Moreover, the $22\pmod{101}$ segment must have opposite orientation from the $7$ and $15\pmod{101}$ segment, and the $3\pmod{101}$ segment must have opposite orientation from the $1$ and $2\pmod{101}$ segment. This suffices to force the construction of the outer frame.

    For the inner frame, note that the $0\pmod{101}$ segment right after the segment of length $1$ must be forced to be that angle (otherwise, it collides with the outer frame), and this subsequently forces all other segments up to the point $101v$. The construction up to point $101u$ is then forced because there is only one possible angle for the path to leave $101u$ avoiding a collision. Hence, the entire frame is forced.
\end{proof}
Finally, it suffices to prove the following.
\begin{claim}
    Let $\mathcal P'$ be the concatenation of the blown-up string $101\mathcal P$ and the frame construction $F$. Then, there is a closed chain satisfying $\mathcal P'$ if and only if there is a chain inside $B$ that goes from $u$ to $v$ satisfying $\mathcal P$.
\end{claim}
\begin{proof}
    For the direction $(\Rightarrow)$, if there is a closed chain satisfying $\mathcal P'$, note that by Claim \ref{claim:frame_is_forced}, $F$ must be folded into the structure shown in the figure above. Since $101\mathcal P$ can travel between points with coordinates $0\pmod{101}$, $101\mathcal P$ must travel inside the bounding box $101B$. Hence, $\mathcal P$ must travel inside the bounding box $B$.
    
    For the direction $(\Leftarrow)$, if $\mathcal P$ travels inside bounding box $B$ from $u$ to $v$, then $101\mathcal P$ must travel inside $101B$ from $101u$ to $101v$. The frame construction as shown in Figure \ref{fig:triangular_frame} then connects $101v$ back to $101u$, forming a closed chain.
\end{proof}
This establishes the NP-hardness of \textsc{Triangular Closed Chain}.

\section{Weak NP-Hardness of Filling a Rectangular Grid}
\label{sec:weak2DFill}

In this section, we consider \textsc{2D Snake Cube}, where the chain must fill a $1 \times H \times W$ rectangle. However, we allow the instructions to be encoded using the shorthand notation, which keeps the inputs polynomial with respect to the input integers. Since this modification means the problem may no longer be in NP, this reduction only proves NP-hardness. For any set $S$, let $\sum S$ be the sum of its elements.

Let $A$ be the multiset of positive integers, a \textsc{2-Partition} instance. We select $H = 2|A| + 4$ and $W = 4 \sum A + 1$. The program comprises the \vocab{caps} at either end and $|A|$ \vocab{layers} in between, encoding each $a_i$ in $A$ sequentially. The \vocab{swivel points} join each gadget and allow the layers to flip horizontally. The orientation of each layer left or right corresponds to assigning each $a_i$ to either partition (Figure \ref{fig:weak2dfill}).
\begin{figure}[htp]
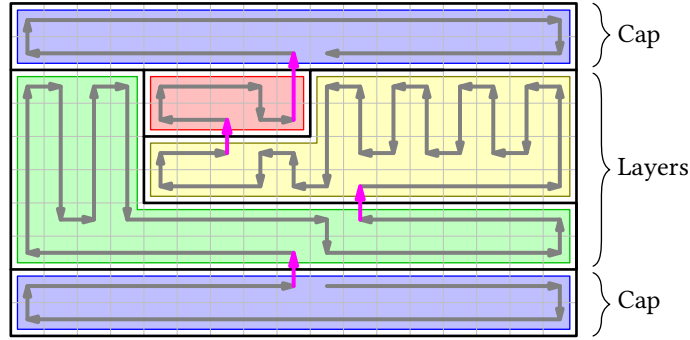

    \centering
    \asyinclude[width=0.55\linewidth]{asy/weak2D_overall.asy}
    \caption{Chain for $A = \{1,2,1\}$, emphasizing the different gadgets, highlighting the swivel points (in magenta), and demonstrating the 3 variants of layers.}
    \label{fig:weak2dfill}
\end{figure}

\subsection{Gadgets}
The starting cap is the subsequence (the ending cap being the reverse): 
\begin{align*}
    (\texttt S)^{\frac{W-1}2 - 1}\texttt{TT}(\texttt S)^{W-2}\texttt{TT}(\texttt S)^{\frac{W+1}2 - 1}\texttt{T} \ldots 
\end{align*}

Since $W > H$, the $W$-segments in the caps can only fit horizontally, forcing the orientation of the entire chain. To avoid creating unfillable empty space, the other segments in the cap must be on the same side of the $W$-segments because the $W$-segments obstruct any other part of the chain from filling it (Figure \ref{subfig:weak2dfill_incorrectcap}). This forces the $W$-segments to be at the top and bottom, and this in turn forces the swivel points joining the caps and the layers to be both horizontally centered and vertically one square from the boundaries, regardless of either horizontal reflection of the caps (Figure \ref{subfig:weak2dfill_correctcap}).

\begin{figure}[htp]
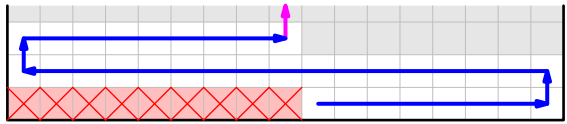
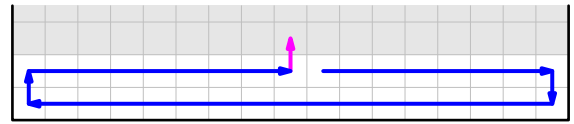

    \centering
    \begin{subfigure}{0.45\textwidth}
        \centering
        \asyinclude[width=\linewidth]{asy/weak2D_incorrectcap.asy}
        \caption{An arrangement that creates unfillable space (red, crossed) regardless of the rest of the chain (gray)}
        \label{subfig:weak2dfill_incorrectcap}
    \end{subfigure}
    \hfill
    \begin{subfigure}{0.45\textwidth}
        \centering
        \asyinclude[width=\linewidth]{asy/weak2D_correctcap.asy}
        \caption{Intended folding allowing the rest of the chain (gray) to fill the remaining space.}
        \label{subfig:weak2dfill_correctcap}
    \end{subfigure}
    \caption{Forced arrangement of cap segments. }
\end{figure}

For each $a_i$, let $A_i = \{a_1,a_2,\dots,a_i \}$, $w_i = 4\sum (A \setminus A_{i-1}) + 1$, $x_i = (w_i - 1)/2$, and $h_i = 2\lvert A \setminus A_{i-1}\rvert$. There are 3 variants of the corresponding layer gadget. For ease of discussion, sections of the gadget are named ( Figure \ref{subfig:weak2dfill_layer}).
\begin{itemize}
\item \textbf{Variant 1:} If $4a_i \leq x_i$ and $h_i > 2$, then the layer is the subsequence:
\begin{align*}
    \ldots\texttt{T}(\texttt S)^{x_i - 1}\texttt{T} && \text{``arm''} \\
    (\texttt S)^{h_i-2}(\texttt{T}\texttt{T}(\texttt S)^{h_i-3})^{4a_i - 1}\texttt{T} && \text{``padding''} \\
    (\texttt S)^{x_i - 4a_i + 1}(\texttt{T})^{2(2a_i-1)} && \text{``shift''} \\
    (\texttt S)^{x_i - 2a_i}\texttt{T}\texttt{T}(\texttt S)^{x_i - 2a_i - 1}\texttt{T}\ldots && \text{``return''}
\end{align*}
The ``arm'' can be though of as shifting the chain to a corner of the remaining space. Similarly, the ``padding'' fills up $4a_i$ units of horizontal space, the ``shift'' fills space to shift the swivel axis, and the ``return'' properly positions the next swivel point in the middle of the remaining space.

\item \textbf{Variant 2:} If $4a_i > x_i$ and $h_i > 2$, informally the padding can be visualized as spilling over into the shift, resulting in these differences:
\begin{align*}
    (\texttt S)^{{h_i}-2}(\texttt{T}\texttt{T}(\texttt S)^{{h_i}-3})^{x_i}(\texttt{T}\texttt{T}(\texttt S)^{{h_i}-2})^{4a_i - x_i - 1}\texttt{T} && \text{``padding''} \\
    (\texttt{T})^{2(2a_i-1 - (4a_i - x_i - 1)) } && \text{``shift''}
\end{align*}
\item \textbf{Variant 3:} If $h_i=2$, informally the padding can be visualized as subsuming the shift entirely, resulting in these changes from the first variant:
\begin{align*}
    \texttt{T}(\texttt S)^{x_i}(\texttt{T})^{2(2a_i-1)} && \text{``padding''} \\ 
    \texttt{T}\texttt S && \text{``return''} 
\end{align*}
\end{itemize}

\begin{figure}[htp]
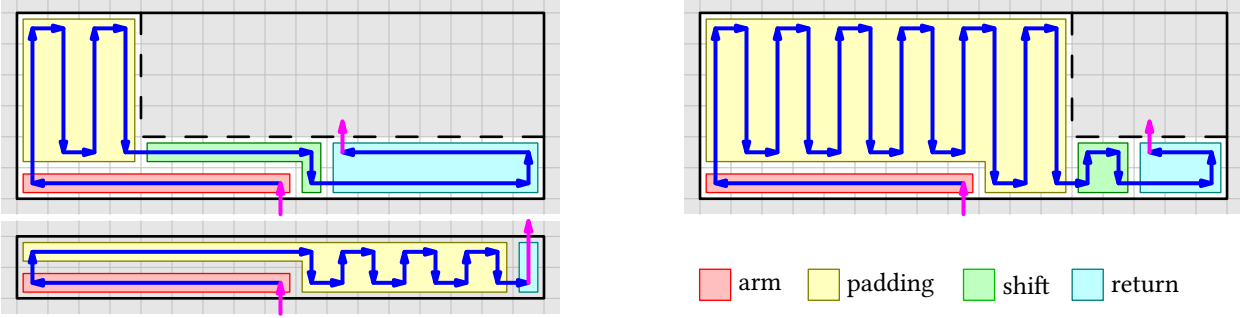

    \centering
    \begin{subfigure}{0.45\textwidth}
        \centering
        \asyinclude[width=\linewidth]{asy/weak2D_layer_variant1.asy}
    \end{subfigure}
    \hfill
    \begin{subfigure}{0.45\textwidth}
        \centering
        \asyinclude[width=\linewidth]{asy/weak2D_layer_variant2.asy}
    \end{subfigure}
    \hfill
    \begin{subfigure}{0.45\textwidth}
        \centering
        \asyinclude[width=\linewidth]{asy/weak2D_layer_variant3.asy}
    \end{subfigure}
    \hfill
    \begin{subfigure}{0.45\textwidth}
        \centering
        \asyinclude[width=\linewidth]{asy/weak2D_layer_key.asy}
    \end{subfigure}
    \caption{Layer gadget mechanics. Top left shows an example of Variant 1, top right shows Variant 2, and bottom left shows Variant 3. For all examples, $w_i = 17$. For the example for Variant 1, $a_i = 1, h_i = 6$; for Variant 2, $a_i = 3, h_i = 6$; for Variant 3, $a_i = 4, h_i = 2$.}
    \label{subfig:weak2dfill_layer}
\end{figure}

Each layer gadget has a $w_i \times h_i$ rectangular space available to it and leaves behind a $w_{i+1} \times h_{i+1}$ rectangular space, while displacing the next swivel point horizontally by $2a_i$ left or right and vertically up by $2$ relative to the previous swivel point, thus keeping the swivel point horizontally centered at the bottom of each space (Figure \ref{subfig:weak2dfill_layer}). The horizontal reflection of this gadget also must have these properties. To show this, we use induction starting from the first layer. Note that the arm and padding sections are all forced by space constraints. The shift section is forced since turning the chain outward at any point in the subsequence of repeated turns ($\texttt{T}$) would leave behind a $1 \times 1$ space (Figure \ref{fig:weak2dfill_layershift}), which can only be filled by the endpoints of the chain since once filled, the chain cannot continue onwards. However, the position of the endpoints is forced by the cap gadgets and cannot fill this space for any of the layer gadgets. Then, the return section is also forced by space constraints. 

\begin{figure}[htp]
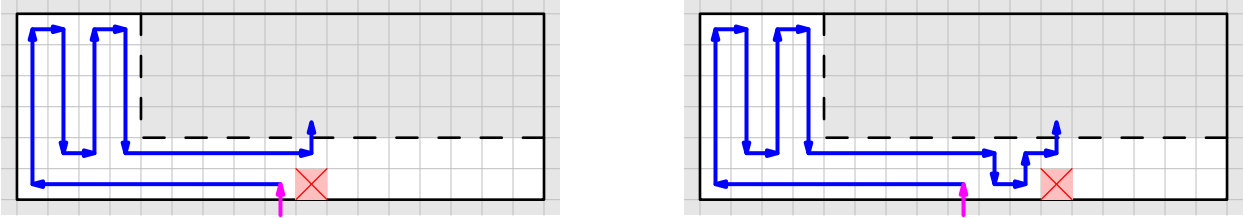

    \centering
    \begin{subfigure}{0.45\textwidth}
        \centering
        \asyinclude[width=\linewidth]{asy/weak2D_layershift1.asy}
    \end{subfigure}
    \hfill
    \begin{subfigure}{0.45\textwidth}
        \centering
        \asyinclude[width=\linewidth]{asy/weak2D_layershift2.asy}
    \end{subfigure}
    \caption{Arrangements of the shift section that leave behind unfillable space.}
    \label{fig:weak2dfill_layershift}
\end{figure}

\subsection{Reduction}

If there exists a solution to \textsc{2-Partition}, then construct all the gadgets and flip the layer gadgets so that arms for all numbers in $A_1$ point to the left, and those for numbers in $A_2$ point to the right. The horizontal displacements of the swivel points must sum to $0$, so the last layer can connect to the upper cap. 

If there exists a solution for \textsc{2D Snake Cube}, then we have demonstrated the gadgets are forced to be constructed in the correct orientation. Since the last layer gadget connects to the upper cap gadget, the horizontal displacements of the swivel points must sum to $0$. Reverse the above process to obtain a solution to the 2-Partition instance. 

\section*{Acknowledgements}
This work was conducted during open problem-solving in the MIT class on Algorithmic Lower Bounds: Fun with Hardness Proofs (6.5440) in Fall 2023. We thank the other participants in the class---particularly Papon Lapate, Benson Lin Zhan Li, and Kevin Zhao---for related discussions and for providing an inspiring atmosphere.

\bibliographystyle{alpha}
\bibliography{bibliography.bib}

\newcommand{\etalchar}[1]{$^{#1}$}
\begin{thebibliography}{ADD{\etalchar{+}}13}

\bibitem[AD11]{foldingpolyhedra}
Zachary Abel and Erik~D. Demaine.
\newblock Edge-unfolding orthogonal polyhedra is strongly {NP}-complete.
\newblock In {\em Proceedings of the 23rd Annual Canadian Conference on
  Computational Geometry (CCCG 2011)}, 2011.

\bibitem[ADD{\etalchar{+}}13]{abel2013finding}
Zachary Abel, Erik~D. Demaine, Martin~L. Demaine, Sarah Eisenstat, Jayson
  Lynch, and Tao~B. Schardl.
\newblock Finding a {Hamiltonian} path in a cube with specified turns is hard.
\newblock {\em Journal of Information Processing}, 21(3):368--377, 2013.

\bibitem[CDBG11]{hampathfilling}
Kenneth~C. Cheung, Erik~D. Demaine, Jonathan~R. Bachrach, and Saul Griffith.
\newblock Programmable assembly with universally foldable strings (moteins).
\newblock {\em IEEE Transactions on Robotics}, 27(4):718--729, 2011.

\bibitem[DE11]{fixedangle_chain}
Erik~D. Demaine and Sarah Eisenstat.
\newblock Flattening fixed-angle chains is strongly {NP}-hard.
\newblock In {\em Proceedings of the 12th Algorithms and Data Structures
  Symposium (WADS 2011)}, pages 314--325, August 15--17 2011.

\bibitem[DILU22]{original2Dsnake}
Erik~D. Demaine, Hiro Ito, Jayson Lynch, and Ryuhei Uehara.
\newblock Computational complexity of flattening fixed-angle orthogonal chains.
\newblock arXiv:2212.12450, 2022.
\newblock \url{https://arXiv.org/abs/2212.12450}. Preliminary version in CCCG
  2022.

\bibitem[GJ90]{garey_johnson}
Michael~R. Garey and David~S. Johnson.
\newblock {\em {Computers and Intractability; A Guide to the Theory of
  NP-Completeness}}.
\newblock W. H. Freeman \& Co., USA, 1990.

\end{thebibliography}

\end{document}